\RequirePackage[l2tabu,orthodox]{nag}
\documentclass
[letterpaper,11pt,]
{article}

\usepackage{etex}
\usepackage{verbatim}
\usepackage{xspace,enumerate}
\usepackage[dvipsnames]{xcolor}
\usepackage[T1]{fontenc}
\usepackage[full]{textcomp}
\usepackage[american]{babel}
\usepackage{mathtools}

\usepackage{graphicx}
\usepackage{pgfplots}
\usepackage{amsmath}

\usepackage{amsthm}
\usepackage[
letterpaper,
top=1in,
bottom=1in,
left=1in,
right=1in]{geometry}
\usepackage{newpxtext} %
\usepackage{textcomp} %
\usepackage[varg,bigdelims]{newpxmath}
\usepackage[scr=rsfso]{mathalfa}%
\usepackage{bm} %
\let\mathbb\varmathbb
\usepackage{microtype}
\usepackage[pagebackref,colorlinks=true,urlcolor=blue,linkcolor=blue,citecolor=OliveGreen]{hyperref}

\usepackage[capitalise,nameinlink]{cleveref}
\crefname{lemma}{Lemma}{Lemmas}
\crefname{fact}{Fact}{Facts}
\crefname{theorem}{Theorem}{Theorems}
\crefname{corollary}{Corollary}{Corollaries}
\crefname{claim}{Claim}{Claims}
\crefname{example}{Example}{Examples}
\crefname{algorithm}{Algorithm}{Algorithms}
\crefname{problem}{Problem}{Problems}
\crefname{definition}{Definition}{Definitions}
\crefname{exercise}{Exercise}{Exercises}
\crefname{condition}{Condition}{Conditions}
\crefname{figure}{Figure}{Figures}
\crefname{observation}{Observation}{Observations}

\newtheorem{theorem}{Theorem}[section]
\newtheorem*{theorem*}{Theorem}

\newtheorem*{lemma*}{Lemma}
\newtheorem{fact}[theorem]{Fact}
\newtheorem*{fact*}{Fact}

\newtheorem*{proposition*}{Proposition}
\newtheorem{corollary}[theorem]{Corollary}
\newtheorem*{corollary*}{Corollary}

\newtheorem*{hypothesis*}{Hypothesis}

\newtheorem*{conjecture*}{Conjecture}
\theoremstyle{definition}
\newtheorem{definition}[theorem]{Definition}
\newtheorem*{definition*}{Definition}

\newtheorem*{construction*}{Construction}

\newtheorem*{example*}{Example}

\newtheorem*{question*}{Question}
\newtheorem{algorithm}[theorem]{Algorithm}
\newtheorem*{algorithm*}{Algorithm}

\newtheorem*{assumption*}{Assumption}

\newtheorem*{problem*}{Problem}

\newtheorem*{openquestion*}{Open Question}
\newtheorem{observation}[theorem]{Observation}
\newtheorem*{observation*}{Observation}
\theoremstyle{remark}

\newtheorem*{claim*}{Claim}

\newtheorem*{remark*}{Remark}
\usepackage{paralist}
\let\originalleft\left
\let\originalright\right
\renewcommand{\left}{\mathopen{}\mathclose\bgroup\originalleft}
\renewcommand{\right}{\aftergroup\egroup\originalright}
\usepackage{turnstile}
\usepackage{mdframed}
\usepackage{tikz}
\usetikzlibrary{positioning}
\usepackage{caption}
\DeclareCaptionType{Algorithm}
\usepackage{newfloat}
\usepackage{float}
\usepackage{array}
\usepackage{booktabs}
\usepackage{subfig}
\usepackage{bbm}
\usepackage{xparse}
\usepackage{amsthm} %
\makeatletter
\let\latexparagraph\paragraph
\RenewDocumentCommand{\paragraph}{som}{%
  \IfBooleanTF{#1}
    {\latexparagraph*{#3}}
    {\IfNoValueTF{#2}
       {\latexparagraph{\maybe@addperiod{#3}}}
       {\latexparagraph[#2]{\maybe@addperiod{#3}}}%
  }%
}
\newcommand{\maybe@addperiod}[1]{%
  #1\@addpunct{.}%
}
\makeatother

\newcommand{\Authornotecolored}[3]{}
\newcommand{\Authorcomment}[2]{}
\newcommand{\Authorfnote}[2]{}

\usepackage{boxedminipage}
\newcommand{\paren}[1]{(#1)}
\newcommand{\Paren}[1]{\left(#1\right)}

\newcommand{\brac}[1]{[#1]}
\newcommand{\Brac}[1]{\left[#1\right]}

\newcommand{\abs}[1]{\lvert#1\rvert}

\newcommand{\set}[1]{\{#1\}}
\newcommand{\Set}[1]{\left\{#1\right\}}

\newcommand{\norm}[1]{\lVert#1\rVert}
\newcommand{\Norm}[1]{\left\lVert#1\right\rVert}

\newcommand{\normt}[1]{\norm{#1}_2}

\newcommand{\iprod}[1]{\langle#1\rangle}
\newcommand{\Iprod}[1]{\left\langle#1\right\rangle}

\newcommand{\Esymb}{\mathbb{E}}

\DeclareMathOperator*{\E}{\Esymb}

\newcommand\bdot\bullet

\DeclareMathOperator{\poly}{poly}

\newcommand{\iid}{i.i.d.\xspace}

\newcommand{\N}{\mathbb N}
\newcommand{\R}{\mathbb R}

\newcommand{\cA}{\mathcal A}
\newcommand{\cB}{\mathcal B}

\newcommand{\cD}{\mathcal D}

\renewcommand{\leq}{\leqslant}

\renewcommand{\geq}{\geqslant}

\let\epsilon=\varepsilon
\numberwithin{equation}{section}
\newcommand\MYcurrentlabel{xxx}
\newcommand{\MYstore}[2]{%
  \global\expandafter \def \csname MYMEMORY #1 \endcsname{#2}%
}
\newcommand{\MYload}[1]{%
  \csname MYMEMORY #1 \endcsname%
}
\newcommand{\MYnewlabel}[2][]{%
  \renewcommand\MYcurrentlabel{#2}%
  \if\relax\detokenize{#1}\relax
    \MYoldlabel{#2}%
  \else
    \MYoldlabel[#1]{#2}%
  \fi
}
\newcommand{\MYdummylabel}[2][]{}
\newcommand{\torestate}[1]{%
  \let\MYoldlabel\label%
  \let\label\MYnewlabel%
  #1%
  \MYstore{\MYcurrentlabel}{#1}%
  \let\label\MYoldlabel%
}
\newcommand{\restatetheorem}[1]{%
  \let\MYoldlabel\label
  \let\label\MYdummylabel
  \begin{theorem*}[Restatement of \cref{#1}]
    \MYload{#1}
  \end{theorem*}
  \let\label\MYoldlabel
}
\newcommand{\restatelemma}[1]{%
  \let\MYoldlabel\label
  \let\label\MYdummylabel
  \begin{lemma*}[Restatement of \cref{#1}]
    \MYload{#1}
  \end{lemma*}
  \let\label\MYoldlabel
}
\newcommand{\restateprop}[1]{%
  \let\MYoldlabel\label
  \let\label\MYdummylabel
  \begin{proposition*}[Restatement of \cref{#1}]
    \MYload{#1}
  \end{proposition*}
  \let\label\MYoldlabel
}
\newcommand{\restatefact}[1]{%
  \let\MYoldlabel\label
  \let\label\MYdummylabel
  \begin{fact*}[Restatement of \cref{#1}]
    \MYload{#1}
  \end{fact*}
  \let\label\MYoldlabel
}
\newcommand{\restateobs}[1]{%
  \let\MYoldlabel\label
  \let\label\MYdummylabel
  \begin{observation*}[Restatement of \cref{#1}]
    \MYload{#1}
  \end{observation*}
  \let\label\MYoldlabel
}
\newcommand{\restate}[1]{%
  \let\MYoldlabel\label
  \let\label\MYdummylabel
  \MYload{#1}
  \let\label\MYoldlabel
}

\newcommand{\sse}{\subseteq}

\newcommand{\e}{\epsilon}
\newcommand{\eps}{\epsilon}

\allowdisplaybreaks
\newcommand*{\normop}[1]{\norm{#1}_{\mathrm{op}}}
\newcommand*{\Normop}[1]{\Norm{#1}_{\mathrm{op}}}

\newcommand*{\normf}[1]{\norm{#1}_{\mathrm{F}}}
\newcommand*{\Normf}[1]{\Norm{#1}_{\mathrm{F}}}

\newcommand{\pE}{\tilde{\mathbb{E}}}

\newcommand{\proves}[2]{\sststile{#1}{#2}}

\newcommand{\NP}{\mathrm{NP}}
\newcommand{\DTIME}{\mathrm{DTIME}}

\newcommand{\normtwofour}[1]{\norm{#1}_{2\rightarrow 4}}
\newcommand{\normtwoq}[1]{\norm{#1}_{2\rightarrow q}}
\newcommand{\normpq}[1]{\norm{#1}_{p\rightarrow q}}

\newcommand{\normq}[1]{\norm{#1}_q}

\newcommand{\normtwofourexp}[1]{\norm{#1}_{2\rightarrow \bar{4}}}
\newcommand{\normtwoqexp}[1]{\norm{#1}_{2\rightarrow \bar{q}}}

\newcommand{\normqexp}[1]{\norm{#1}_{\bar{q}}}
\newcommand{\normfourexp}[1]{\norm{#1}_{\bar{4}}}

\newcommand{\normtwoinfty}[1]{\norm{#1}_{2 \rightarrow \infty}}
\newcommand{\normtwotwoexp}[1]{\norm{#1}_{2 \rightarrow \bar{2}}}
\newcommand{\normtwotwo}[1]{\norm{#1}_{2 \rightarrow 2}}

\newcommand{\QMAt}{\mathrm{QMA}(2)}
\newcommand{\EXP}{\mathrm{EXP}}

\newcommand{\ALG}{\mathrm{ALG}}

\def\colorful{0}

\ifnum\colorful=1

\newcommand{\note}[3]{\textcolor{#1}{#2: #3}}
\newcommand{\stefan}[1]{\note{blue}{[Stefan]}{#1}}
\newcommand{\stefantopic}[1]{\note{ForestGreen}{[Stefan]}{#1}}
\newcommand{\sam}[1]{\note{orange}{[Sam]}{#1}}
\newcommand{\samtopic}[1]{\note{ForestGreen}{[Sam]}{#1}}
\else

\newcommand{\note}[3]{}
\newcommand{\stefan}[1]{}
\newcommand{\stefantopic}[1]{}
\newcommand{\sam}[1]{}
\newcommand{\samtopic}[1]{}
\fi

\title{Algorithms with Polynomially-Improved Approximation Factors for the $2 \rightarrow q$ Norm, and Applications }
\newcommand{\authorinfo}[4]{%
  \begin{tabular}[t]{c}
    #1\thanks{#2}\\
    #3\\
    \texttt{#4}
  \end{tabular}%
}

\author{
  \authorinfo{Samuel B. Hopkins}{Supported by NSF award no. 2238080, MLA@CSAIL, FinTechAI@CSAIL, MIT Research Support Committee, MIT-Google Program for Computing Innovation.}{samhop@mit.edu}{MIT}
  \and
  \authorinfo{Stefan Tiegel}{Supported by an SNF Postdoc.Mobility Grant (P500-2 235374).}{stefan23@mit.edu}{MIT}
}
\date{\today}

\date{}

\begin{document}
\maketitle

\thispagestyle{empty}

\begin{abstract}
The $2 \rightarrow q$ norm of a matrix $X \in \R^{n \times d}$ is defined as $\normtwoq{X} = \sup_{\normt{v} = 1} \normq{Xv}$.
We give polynomial-time multiplicative approximation algorithms for this norm when $q > 2$ (i.e. in the hypercontractive setting).
This problem either directly captures or is closely related to long-standing open problems in combinatorial optimization and hardness of approximation (e.g. Small Set Expansion), quantum information (e.g. Best Separable State), and algorithmic statistics.

Very little is known about what approximation factors we can achieve for this problem in polynomial time, even though such approximations have significant downstream consequences.
Barak, Brand\~{a}o, Harrow, Kelner, Steurer, and Zhou showed that no polynomial-time algorithm can achieve an approximation factor better than $2^{\sqrt{\log n}}$, assuming the Exponential Time Hypothesis \cite{barak2012hypercontractivity}.
On the other hand, a simple spectral algorithm gives a $d^{1/4}$-approximation as a baseline.
For the important special case of $q = 4$, prior work of Guth, Maldague, and Urschel \cite{guth2025estimating} can be combined with known sparsification techniques to give a $d^{1/6}$-approximation. We improve over these results by polynomial factors, giving a $d^{1/8}$-approximation.

Moreover, we construct \emph{sum-of-squares} certificates for the $2 \rightarrow q$ norm.
This directly implies improved algorithms for robust mean and covariance estimation, robust regression, and clustering, when the data only satisfies a bound on its $q$-th moment.
\end{abstract}

\clearpage

\thispagestyle{empty}
\microtypesetup{protrusion=false}
\tableofcontents{}
\microtypesetup{protrusion=true}

\clearpage

\pagestyle{plain}
\setcounter{page}{1}

\section{Introduction}
\label{sec:intro}
\stefantopic{Definitions}
We study approximation algorithms for the $2 \rightarrow q$ norm of matrices.
For $q \geq 1$, the $2 \rightarrow q$ norm of a matrix $X \in \R^{n \times d}$ is defined as
\[
    \normtwoq{X} = \sup_{\norm{v}_2 = 1} \Norm{Xv}_q \,.
\]
Special cases are the spectral norm ($q = 2$) and the largest $\ell_2$-norm of the rows of $X$ ($q = \infty$).
We focus on the in-between setting of $2 < q < \infty$, known as the \emph{hypercontractive} setting.
For such $q$, the $2 \rightarrow q$ norm captures how `spiky' or `analytically sparse' vectors generated by the columns of $X$ are.

While the $2 \rightarrow 2$ and $2 \rightarrow \infty$ norms can be computed efficiently, even approximating the $2 \rightarrow q$ norm for general $2 < q < \infty$ appears to be significantly more challenging \cite{barak2012hypercontractivity,harrow2019limitations,bhattiprolu2023inapproximability}.
Inspired by \cite{barak2012hypercontractivity}, we say an algorithm achieves a $\gamma$-approximation if it can solve the following promise problem.\footnote{In \cite{barak2012hypercontractivity} the right-hand side scales with the minimum non-zero singular value of $X$ to make the definition scale-invariant. For the applications we mention, the version we consider is sufficient. \cite[Lemma 9.7]{barak2012hypercontractivity} shows a reduction from our setting to their setting.}
Given $X$, decide whether it satisfies
(\textbf{YES case}) $\normtwoq{X} \geq \gamma$ or (\textbf{NO case}) $\normtwoq{X} \leq 1$.

Such approximation algorithms have direct consequences for a number of important problems such as Small Set Expansion and Best Separable State and key questions in algorithmic statistics \cite{barak2012hypercontractivity,harrow2013testing,hopkins2019hard}.
Despite its centrality, very little is known about what approximation factors we can achieve in polynomial time.
There is a `naive' approximation algorithm based on a simple spectral certificate involving the rows of $X$ achieving an approximation factor of $d^{1/4}$.
Specializing for a moment on the important special case of $q = 4$, prior work of Guth, Maldague, and Urschel \cite{guth2025estimating} can be combined with known sparsification techniques (described after \cref{thm:main_intro}) to give a $d^{1/6}$-approximation.
We improve over these results by polynomial factors, giving a $d^{1/8}$-approximation.
(Prior works and ours extend to larger even $q$, with approximation factor depending on $q$.)

Our algorithm is based on computing eigenvalues of carefully-constructed family of matrices.
Further, it is also captured by the low-degree sum-of-squares proof system, which is important for the applications to algorithmic statistics.

\paragraph{Algorithmic Statistics, Small Set Expansion, and QMA[2]}
\stefantopic{General motivation for the problem}
Improvements to the approximation factor have direct consequences to algorithmic statistics, and we detail the consequences of our $d^{1/8}$-approximation below.
Strong enough further improvements would also have significant implications 
for Small Set Expansion and quantum complexity theory.
In particular, a constant-factor approximation with the appropriate run-time refute the Small Set Expansion Hypothesis of Raghavendra and Steurer \cite{raghavendra2010graph} (under standard complexity theoretic assumptions), a central conjecture in hardness of approximation \cite{raghavendra2008optimal,raghavendra2010graph,raghavendra2010approximations,arora2010subexponential,raghavendra2012reductions,barak2014sum,bafna2021playing,ghoshal2023lifting}.
Weaker approximation factors still yield interesting approximation guarantees.
A constant-factor approximation would additionally resolve the relation of the quantum complexity class $\QMAt$ to classical complexity classes~\cite{harrow2013testing}, answering an important problem in quantum complexity about the power of two unentangled provers \cite{harrow2013testing,barak2012hypercontractivity}.\footnote{In particular, a quasi-polynomial time algorithm would imply that $\QMAt \sse \EXP$.}
We describe the relation to algorithmic statistics next and defer a discussion about Small Set Expansion to after \cref{thm:main_intro}.

\vspace{0.5em}

\textbf{Algorithmic statistics.}
If $x_1,\ldots,x_n$ are the rows of $X$, the $2 \rightarrow q$ norm of $X$ is equivalently the maximum $q$-th moment of any one-dimensional projection of $x_1,\ldots,x_n$, up to normalization.
Constant-factor approximation algorithms for this quantity are at the core of many recent algorithms for challenging problems in algorithmic statistics, such as robust estimation of mean and covariance, robust linear regression, and clustering \cite{hopkins2018mixture,kothari2018robust,klivans2018efficient,bakshi2021robust,steurer2021sos,gollakota2023tester}.
Currently, such approximations are only known when $x_1,\ldots,x_n$ are independent draws from a `nice' probability distribution, e.g. Gaussian, sub-Gaussian, or satisfying a Poincar\'e Inequality \cite{kothari2018robust,hopkins2018mixture,diakonikolas2025sos}.
Whether this type of average-case assumption on $x_1,\ldots,x_n$ can be weakened or removed is thus also central open question in algorithmic statistics \cite{hopkins2019hard,diakonikolas2025sos}.
A positive answer would yield efficient algorithms for core problems in robust statistics and clustering under almost information-theoretically minimal distributional assumptions.
Similarly, weaker yet non-trivial approximation factors directly lead to improved estimation guarantees under significantly weaker distributional assumptions than current efficient algorithms require.

\stefantopic{Previous works and main question of our work}
\paragraph{Best approximation factor in polynomial time?}
Despite significant investigations of various special cases and variations \cite{bhaskara2011approximating,barak2012hypercontractivity,bhattiprolu2017weak,harrow2019limitations,bhattiprolu2023inapproximability,guth2025estimating,bhattiprolu2025inapproximability}, our understanding of the computational complexity of approximating the $2 \rightarrow q$ norm is very limited.
The typical regime of interest is when $n \gtrsim d^{q/2}$.
The only known constant-factor approximation algorithm is based on brute-force search in the column span of $X$, requiring time $\exp(\Omega(d))$.
Further, \cite{bhattiprolu2017weak} give sub-exponential (in $d$) algorithms achieving an $O_q(d^{(1-\alpha)(1/2 - 1/q)})$-approximation in time $\exp(\tilde{O}(d^{\alpha}))$ for $\alpha \in [0,1]$.
On the other hand, \cite{barak2012hypercontractivity} show that approximating the $2 \rightarrow 4$ norm to within a factor of $1 + n^{-O(1)}$ is $\mathrm{NP}$-hard and that a $2^{o(\sqrt{\log n})}$-approximation requires time at least $2^{\omega(\log n)}$ under the Exponential Time Hypothesis.\footnote{More specifically, they show that for any $\e, \delta > 0$ such that $2\e + \delta < 1$, obtaining an $\exp(\log^\e(n))$-approximation factor takes time at least $\exp(\log^{(1+\delta)}(n))$ under ETH.}
Further, \cite{harrow2019limitations} (unconditionally) showed that approximating the $2\rightarrow4$ norm up to any constant needs at least $\Omega(\log d/ \poly(\log \log d))$ rounds of the sum-of-squares hierarchy.

This state of affairs in particular leaves open the following tantalizing possibility:
For any $\eps > 0$, can we achieve an $d^{\e}$-approximation in polynomial time?\footnote{Similarly, a constant-factor approximation in quasi-polynomial time is not ruled out.}
An affirmative answer would be a major breakthrough, implying similarly strong approximations for Small Set Expansion.
In this work we thus work toward the following question:
\begin{center}
    \emph{What is the best approximation factor we can achieve in polynomial time (in $n$ and $d$)?}
\end{center}

\stefantopic{Baseline + ETH}
Focusing on $q = 4$ for a moment, the ETH hardness result of \cite{barak2012hypercontractivity} shows that for the $2 \rightarrow 4$ norm this needs to be at least $2^{\sqrt{\log n}}$.
Further, we show in \cref{sec:tech_overview} that for even $q$ a baseline for this question is as follows:
The spectral norm of $\tilde{X} = \sum_i x_i^{\otimes (q/2)}[x_i^{\otimes (q/2)}]^\top$ always satisfies $\normtwoq{X} \leq \normop{\tilde{X}}^{1/q} \leq d^{1/4} \cdot \normtwoq{X}$, thus giving a $d^{1/4}$-approximation algorithm.
The recent work of Guth, Maldague, and Urschel combined with known sparsification techniques, described below, achieves a $d^{1/6}$-approximation \cite{guth2025estimating} for $q = 4$.

\stefantopic{Statement of main result}
\paragraph{Our contribution}
We give a polynomial-time algorithm that improves over both of the above algorithms by polynomial factors, both our and previous works extend to larger even $q$ as well.
Our algorithm also applies to the search version of the problem.
We also observe that a simple tensoring/parallel repetition argument combined with the results of \cite{harrow2019limitations} shows that for any $p$, the natural degree-$p$ sum-of-squares certification algorithm can at best achieve an $d^{\Omega(1/p)}$-approximation.
\begin{theorem}[See \cref{thm:two_to_q_full} for the full version]
    \label{thm:main_intro}
    Let $q \in \N$ be even.
    There is an $O(n^2d^2 + nd^3)$-time\footnote{Optimizing the exact polynomial runtime is not the focus of this work, however, we emphasize that the runtime is much faster than computing the sum-of-squares certificate given by \cref{thm:sos_moment_intro}.} algorithm that approximates the $2 \rightarrow q$ norm of a matrix $X$ to within a factor of $d^{1/4 - 1/(2q)}$.
    Moreover, this algorithm computes a unit vector $\hat{v}$ such that $\normq{X \hat{v}} \geq \normtwoq{X} / d^{1/4 - 1/(2q)}$. 
\end{theorem}

\stefantopic{brief description of list of proxies and SoS certificate}
For the important special case of $q=4$ we achieve a $d^{1/8}$-approximation.
We also show that this factor is best-possible for our approach.
It is an interesting open question to see if variations of it can achieve better approximation factors; we refer to \cref{sec:open_questions}.
Further, explained in more detail in \cref{sec:applications_intro}, we show that in the \textbf{NO} case (when the $2\rightarrow q$ norm is at most 1), there is a separate certificate (certifying an upper bound of $d^{1/4 - 1/(2q)}$) in the form of a low-degree \emph{sum-of-squares proof}.
This is important for the applications to algorithmic statistics, and a certificate of a different form or a generic search or approximation algorithm is not known to lead to improvements for these downstream applications.

\paragraph{Previous polynomial-time algorithms}
Recall that the `trivial' algorithm achieves a $d^{1/4}$-approximation, this holds for all even $q \geq 4$.
For $q = 4$, a simple interpolation argument, described by Steinberg \cite{steinberg2005computation},  using the spectral norm, $\normtwotwo{X}$, and the maximum row norm, $\normtwoinfty{X}$, of $X$ gives an $n^{1/8}$-approximation, together with Lewis-weight sparsification, described next, this also achieves an $d^{1/4}$-approximation.

Recent work of Guth, Maldague, and Urschel achieves an $n^{1/12}$-approximation for $q = 4$ \cite{guth2025estimating}.
Typically, the applications have $n \gtrsim d^{q/2}$, i.e., $n \gtrsim d^2$ for $q = 4$, with larger $n$ possible.
However, for large $n$, $2\rightarrow q$ norm instances can be sparsified using Lewis weights \cite{cohen2015lp} to reduce to the case when $n = \tilde{O}(d^{q/2})$.\footnote{Specifically, we can in polynomial time compute a matrix $X' \in \R^{N\times d}$ for $N = O(d^{q/2} \log d) $ such that with high probability for all vectors $v$, $\normq{Xv}$ and $\normq{X'v}$ are the same up to a constant factor. For even $q$ one can instead use the PSD sparsifiers of Batson, Spielman, and Srivastava~\cite{doi:10.1137/130949117} on a tensored instance to get a deterministic reduction of the number of rows to $N = O(d^{q/2})$.}
Combining this sparsification with \cite{guth2025estimating} yields a $d^{1/6}$-approximation for $q = 4$, and a $d^{(q-2)/4(q-1)}$-approximation for general $q$ (versus our $d^{(q-2)/4q}$-approximation).\footnote{We thank reviewers for FOCS 2026 for pointing out this sparsification trick to us.}

As mentioned above, the literature on algorithmic statistics has (implicitly) given approximation algorithms for average-case versions of this problem, where we additionally assume that in the \textbf{NO} case the rows of $X$ are $n$ \iid samples from some `nice' distribution.
Such `nice' distributions include: (affine transformations of) Gaussians, spherically symmetric, product, \cite{kothari2018robust,hopkins2018mixture}, or sub-Gaussian distributions \cite{diakonikolas2025sos}, or distributions satisfying a Poincar\'e Inequality \cite{kothari2018robust}.

Lastly, \cite{barak2012hypercontractivity} showed that $\normtwofour{X}^4$ can be approximated to within an \emph{additive} factor of $\e \normtwotwo{X}^2 \normtwoinfty{X}^2$ in time $n^{O(\log (n) / \e^2)}$.
\cite{brandao2015estimating} improved this to work for any even $q$ with additive approximation factor $\e \normtwotwo{X}^2 \normtwoinfty{X}^{q-2}$ and runtime $n^{O(1 / \e^2)}$.

\paragraph{Limitations for sum-of-squares}

The work of \cite{harrow2019limitations} shows that the natural sum-of-squares relaxation of the $2 \rightarrow 4$ norm requires at least $\log d / \poly(\log \log d)$ rounds to obtain a constant factor approximation.
Combining this with a tensoring/parallel repetition argument from \cite{bhattiprolu2023inapproximability} leads to the following observation (see \cref{sec:tensoring_sos_lower_bound} for an argument).
\begin{observation}
\torestate{
\label[observation]{obs:SoS_lower_bound}
For any $p \in \N$ sufficiently large, an $d^{1/p}$-approximation to the $2\rightarrow 4$ norm requires at least $p /\poly(\log p)$ rounds of sum-of-squares.
}
\end{observation}

\paragraph{Relation to Small Set Expansion}

Distinguishing graphs that are small set expanders from those that are not reduces to approximating the $2\rightarrow 4$ norm of the Eigenspace (represented e.g. via an orthonormal basis) of Eigenvalues larger than, say, 0.1 of the normalized adjacency matrix.
In particular, a $\gamma$-approximation for the $2\rightarrow 4$ norm allows one to distinguish graphs in which there is a non-expanding set of measure at most $1/\gamma^4$ from graphs in which all sets of at most small constant measure are well-expanding~\cite{barak2012hypercontractivity,barak2014rounding}.\footnote{An analogous reduction between the search problems also exists.}
As a consequence, \cref{thm:main_intro} yields an efficient algorithm when $\delta = d^{-1/2}$, where $d$ is the dimension of the `large' Eigenspaces and $n$ is the size of the graph.
However, the work \cite{barak2014rounding} already solves this problem even when $\delta = d^{-1/3}$.
Their algorithm exploits additional properties of the reduction from Small Set Expansion to the $2 \rightarrow q$ norm problem.

\subsection{Corollaries for Algorithmic Statistics and General Hypercontractive Norms}
\label{sec:applications_intro}

We next describe corrollaries for algorithmic statistics and approximating general hypercontractive norms.

\paragraph{Algorithmic statistics}

A key aspect of the application of $2 \rightarrow q$ norm algorithms, more specifically \emph{certificates} as described below, to algorithmic statistics is its universality:
existing algorithms for \emph{many} problems rely on such certificates in a black-box way \cite{hopkins2018mixture,kothari2018robust,klivans2018efficient,bakshi2021robust,steurer2021sos,gollakota2023tester}, if the certificates can be phrased as sum-of-squares proofs.
Consequently, improvements on the certification question via sum-of-squares algorithms directly lead to improved algorithms for all of these problems.

We pause to highlight one of them, robust mean estimation.
Let $\cD$ be a distribution over $\R^d$ with mean $\mu$ and $q$-th moments bounded along every direction.
Instead of directly receiving \iid samples from $\cD$ as input, an adversary can first inspect this \iid sample and arbitrarily change an $\e$-fraction of the points in it.
We receive the \emph{corrupted} set as input and our goal is to estimate $\mu$ in $\ell_2$-norm. 
Statistically, this problem is well-understood:
Given sufficiently many samples, we can estimate $\mu$ up to dimension-free error $O_q(\e^{1-1/q})$ and this is best-possible \cite{kothari2018robust}.
However, obtaining \emph{efficient} estimators is significantly more challenging and only sub-optimal error guarantees are known:
Applying a one-dimensional estimator coordinate-wise leads to error $O(\e^{1-1/q} \sqrt{d})$ and using exact certificates for bounded second moments (aka the $2\rightarrow 2$ norm) leads to error $O(\sqrt{\e})$ \cite{diakonikolas2017being,steinhardt2018resilience}.

Over the past decade, significant research has been expended to understand the correct scaling of the error as a function of both $\e$ and $d$ in this and closely related settings \cite{diakonikolas2016robust,lai2016agnostic,diakonikolas2017statistical,hopkins2018mixture,kothari2018robust,hopkins2019hard,diakonikolas2020outlier,diakonikolas2022robust,diakonikolas2023algorithmic,diakonikolas2025sos}.
This has led to remarkable progress under additional assumptions (such as $\cD$ being sub-Gaussian).
However, without such assumptions, it is not known how to achieve error better than $O(\min\set{\sqrt{\e},\e^{1-1/q} \sqrt{d}})$.\footnote{One exception where this is achievable is when the covariance of the inlier distribution is \emph{known} to be the identity \cite{diakonikolas2020outlier}.}
A consequence of our results is the first efficient algorithm to break this barrier for a non-trivial range of $\e$.

\vspace{0.5em}

\textbf{Relation to the $2 \rightarrow q$ norm.}
\stefantopic{Transition 2-to-q-norm to moment certificates}
As mentioned before, if $x_1, \ldots, x_n$ are rows of $X$, the $2 \rightarrow q$ norm of $X$ is, up to normalization, equivalent to the maximum $q$-th moment of any one-dimensional projection of the uniform distribution over $x_1, \ldots, x_n$.
\cref{thm:main_intro} shows that when this maximum is at most 1, we can \emph{certify} that it is at most $d^{1/4 - 1/(2q)}$.
For the applications in algorithmic statistics however, it is important that this certificate comes in the form of a \emph{sum-of-squares} proof.
We show that our certificates can be adjusted to be of this form.\footnote{Algorithmically constructing sum-of-squares certificates requires solving a large semi-definite program which has a large polynomial runtime. In contrast, the approximation algorithm of \cref{thm:main_intro} is based on a `proxy list certificate' that  can be computed much faster. It is an interesting open question whether existing algorithms can be adapted to work with these `fast' certificates too to yield algorithms with faster runtime.}
In the following we consider distributions that are uniform over a set of points $x_1, \ldots, x_n \in \R^d$ (e.g., the uncorrupted points for robust mean estimation), this suffices for all applications.
Our results naturally extend to arbitrary distributions, we refer to \cref{sec:sos_moment_technical}.

We say that such a distribution with covariance $\Sigma$ is $(B_q,q)$-bounded/hypercontractive if for any vector $v$, $\E_i \iprod{x_i,v}^q \leq B_q^q \normt{v}^q$ or $\E_i \iprod{x,v}^q \leq B_q^q \iprod{v,\Sigma v}^{q/2}$, respectively, where $\E_i$ denotes the uniform distribution over $[n]$.
Similarly, we say such a distribution is \emph{certifiably} $(B_q,q)$-bounded/hypercontractive when the expression $B_q^q \normt{v}^q - \E_{x \sim \cD} \iprod{x,v}^q$ or $B_q^q \iprod{v,\Sigma v}^{q/2} - \E_{x \sim \cD} \iprod{x,v}^q$ is a sum of squares (of degree $O(q)$) when viewed as a polynomial in formal variables $v$.

\stefantopic{main theorem and discussion}
In these terms, the approximation factor for the $2 \rightarrow q$ norm corresponds to the blow-up in the $B_q$ parameter when going from bounded/hypercontractive to \emph{certifiably} bounded/hypercontractive.
The performance guarantee of efficient algorithms for algorithmic statistics in turn depends on this blow-up.
Understanding what blow-up is necessary is an important question in the area (reiterated in, e.g., \cite{hopkins2019hard,diakonikolas2025sos}).
We show that the proof of \cref{thm:main_intro} can be adapted to yield the following.
\begin{theorem}[See \cref{thm:sos_moment_technical} for the full version]
    \label{thm:sos_moment_intro}
    Let $q \in \N$ be even.
    Any distribution that is $(B_q, q)$-bounded/hypercontractive is $(d^{1/4 - 1/(2q)} B_q, q)$-certifiably bounded/hypercontractive.
\end{theorem}

\cref{thm:sos_moment_intro} immediately implies improvements for many problems in (robust) algorithmic statistics.
For robust mean estimation, \cref{thm:sos_moment_intro} implies that we can robustly estimate the mean of a data set which has bounded fourth moments up to error $O(\e^{3/4}d^{1/8})$.
This beats the previous barrier of $O(\min\set{\sqrt{\e}, \e^{3/4} \sqrt{d}})$ in the range $O(\e \leq d^{-1/2})$.
Similar consequences follow for several other problems, summarized in \cref{table:summary}.

We refer to the listed references for the exact definitions of the models, but give a brief overview in \cref{sec:alg_stats_models}.
`Distributional Assumption' refers to the uniform distribution of the (uncorrupted) input data set.\footnote{Instead, we can also assume that the (uncorrupted) input consists of \iid samples from a distribution with slightly stronger concentration properties, such as bounded $q+\e$ moments for any $\e > 0$.}
The `Previous Best Algorithms' are based on certificates of second moments, i.e., correspond to approximating the $2 \rightarrow 2$ (i.e., the spectral norm), which can be computed exactly.
The input sample in all algorithms below is of size $d^{O(q)}$ and the runtime of order $(n \cdot d)^{O(q)}$.
We hide multiplicative factors depending only on $q$.

For robust linear regression and mean and covariance estimation in the `relative' error metric, we are not aware of any algorithm that did
not have a polynomial dependence on the condition number of the covariance, which could scale
exponentially with the dimension.

\begin{table}[t]
\centering
\caption{Improved algorithms for algorithmic statistics.}
\label{table:summary}
\begin{tabular}{@{}lcccc@{}}
\toprule
Estimation Task                                                                        & \begin{tabular}[c]{@{}c@{}}Distributional\\Assumption\end{tabular} & \begin{tabular}[c]{@{}c@{}} New Guarantees \\ from \Cref{thm:sos_moment_intro}\end{tabular} & \begin{tabular}[c]{@{}c@{}}Previous Best\\ Guarantee in \\ Polynomial Time\end{tabular} 
& \begin{tabular}[c]{@{}c@{}}Best Possible \\ Error\end{tabular} \\ \midrule
\addlinespace[0.7em]
\begin{tabular}[c]{@{}l@{}}Robust mean estimation: \\ Euclidean norm\end{tabular}               & $(1,q)$-bounded                    & \begin{tabular}[c]{@{}c@{}}
$d^{1/4-1/(2q)} \cdot \eps^{1-1/q}$  \\ \scriptsize\cite{hopkins2018mixture,kothari2018robust}\end{tabular}                                  & 
            
           \begin{tabular}[c]{@{}c@{}}$\min\Set{\sqrt{\eps}, \e^{3/4} \sqrt{d}}$ \\ \scriptsize\cite{steinhardt2018resilience,diakonikolas2017being}
           \end{tabular}
                                                                                  &
                         $\e^{1-1/q}$                                   \\ 
\addlinespace[0.7em]
\begin{tabular}[c]{@{}l@{}}List-decodable \\ Mean estimation\end{tabular}          & $(1,q)$-bounded                & \begin{tabular}[c]{@{}c@{}}
$ \frac{d^{1/4 - 1/(2q)}}{\alpha^{2/q}}$  \\ \scriptsize\cite{kothari2018robust}\end{tabular}                                            & 
\begin{tabular}[c]{@{}c@{}}$\frac{1}{\sqrt{\alpha}}$ \\ \,\,\,\,\scriptsize\cite{charikar2017learning}  \end{tabular}
                                                                        &                   \\ 
\addlinespace[0.7em]
\begin{tabular}[l]{@{}l@{}}Mixture models:  \\
Mixture of $k$\\
$\Delta$-separated \\
components
\end{tabular} &   \begin{tabular}[c]{@{}l@{}}Each component \\is   $(1,q)$-bounded                                                    \end{tabular}  
& 
\begin{tabular}[c]{@{}c@{}}
$\Delta \gtrsim d^{1/4 - 1/(2q)} \cdot k^{2/q}$  \\ \scriptsize\cite{hopkins2018mixture,kothari2018robust}\end{tabular}                                                            &                                     \begin{tabular}[c]{@{}c@{}}$\Delta \gtrsim \sqrt{k}$ \\ \,\,\,\,\scriptsize\cite{diakonikolas2022clustering} \end{tabular}
                                                 & \\ 
\addlinespace[1em]
\begin{tabular}[c]{@{}l@{}}Robust mean estimation: \\ {Mahalanobis} norm\end{tabular}             & \begin{tabular}[c]{@{}c@{}}$(1,q)$-\\hypercontractive\end{tabular}                                          & \begin{tabular}[c]{@{}c@{}}
$d^{1/4 - 1/(2q)}\cdot \eps^{1-1/q}$  \\ \scriptsize\cite{hopkins2018mixture,kothari2018robust}\end{tabular}                                      & \begin{tabular}[c]{@{}c@{}} No general \\  algorithm\end{tabular}                                                                        & $\e^{1-1/q}$               \\ 
\addlinespace[0.7em]
\begin{tabular}[c]{@{}l@{}}Covariance estimation: \\ Relative spectral \\
norm\end{tabular}          & \begin{tabular}[c]{@{}c@{}}$(1,q)$-\\hypercontractive\end{tabular}      & \begin{tabular}[c]{@{}c@{}}
$ d^{1/2 - 1/q} \cdot \eps^{1-2/q}$  \\ \scriptsize\cite{kothari2018robust}\end{tabular}                                            & \begin{tabular}[c]{@{}c@{}} No general \\  algorithm\end{tabular}                                                                        & $\e^{1-2/q}$                                                  \\
\addlinespace[0.7em]
\begin{tabular}{@{}l@{}}Robust linear regression:\\ Arbitrary noise with \\ unit scale \end{tabular} &           \begin{tabular}[c]{@{}c@{}}$(1,q)$-\\hypercontractive\end{tabular}                                        & \begin{tabular}[c]{@{}c@{}}
$d^{1/4 - 1/(2q)} \cdot \eps^{1 - \frac{2}{q}}$  \\ \scriptsize\cite{bakshi2021robust}\end{tabular}                                                            &
\begin{tabular}[c]{@{}c@{}} No general \\  algorithm\end{tabular}
& $\eps^{1-2/q}$                  \\ \bottomrule
\end{tabular}
\end{table}

\paragraph{Corollaries for general hypercontractive norms}
\stefantopic{Corolloray for $p \rightarrow q$ norms for $2 < p < q$}

Our results naturally extend to other hypercontractive norms.
In particular, for $p \geq 1$ we can define $\normpq{X} = \sup_{\norm{v}_p = 1} \normq{Xv}$.
Using standard interpolation arguments, \cref{thm:main_intro} implies the following.
\begin{corollary}[See \cref{thm:p_to_q_full} for the full version]
    \label{cor:p_to_q_intro}
    Let $q \in \N$ be even and $1 \leq p \leq q$.
    Let $\gamma_p = 1/p - 1/2$ when $p \leq 2$ and $\gamma_p = 1/2 - 1/p$ when $p \geq 2$.
    There is a polynomial-time algorithm that approximates the $p \rightarrow q$ norm to within a factor of $d^{\gamma_p + 1/4 - 1/(2q)}$.
\end{corollary}

For $p \geq 2$ the previous-best approximation factor was $n^{(q-2)/(2q[q-1])}d^{\gamma_p}$, which is $d^{\gamma_p + 1/4 - 1/(4[q-1])}$ when $n \approx d^{q/2}$ \cite{guth2025estimating}.
Different from $p = 2$, as long as $2 \not\in [p,q]$ approximating the $p \rightarrow q$ norm to within a factor of $2^{\log^{1-\e} n}$ for any $\e$ constant factor is impossible unless $\mathrm{NP} \sse \mathrm{BPTIME}(2^{\log^{O(1)}n})$ \cite{bhattiprolu2023inapproximability}.

\subsection{Related Works}
    The general $p \rightarrow q$ norm problem, beyond the hypercontractive case, has received a lot of attention.
    We list some relevant works and refer to \cite{bhattiprolu2023inapproximability,guth2025estimating} for further references.

    \paragraph{Other works on $p \rightarrow q$ norms}
    In this work we are chiefly concerned with approximating the $2 \rightarrow q$ norm, for other applications of hypercontractive norms in theoretical computer science we refer to \cite{biswal2011hypercontractivity}.
    
    \stefantopic{Non-hypercontracive Regime}
    Approximating the $p \rightarrow q$ norm in the non-hypercontractive case, when $p \geq q$, is much better understood. When $2 \in [q,p]$ the $p \rightarrow q$ norm can be approximated to within a constant factor of at most $\pi/2$ \cite{nesterov1998semidefinite,steinberg2005computation}.
    On the other hand, when $2 \not\in [q,p]$, \cite{bhaskara2011approximating} showed that obtaining any constant-factor approximation is $\NP$-hard, and approximating it to a factor of $2^{\log^{1-\e}n}$ for any constant $\e > 0$ is impossible unless $\NP \sse \DTIME(2^{\log^{O(1)} n})$.
    The works of \cite{steinberg2005computation,guth2025estimating} also give algorithms with polynomial approximation factors in this regime.
    When the entries of the matrix are non-negative, the $p \rightarrow q$ norm can be computed \emph{exactly} for any $p \geq q \geq 1$ \cite{steinberg2005computation,bhaskara2011approximating}.

    \paragraph{Other related works}
    A key part in our argument is to consider `partial contractions' of the fourth moment tensor of the uniform distribution over the rows of $X$.
    Ideas of this form have appeared in several other contexts in theoretical computer science (for instance in \cite{khot2007linear}) -- surveying them all is beyond our scope here.
    Most closely related to our work is the use of such fourth-moment contractions in average-case or semin-random learning problems, especially dictionary learning -- see, e.g., \cite{arora2015simple,awashti2018sparse}.
    The techniques to analyze the contractions in these settings deviate from the setting of worst case matrices $X$ considered in this work.

\section{Technical Overiew}
\label{sec:tech_overview}
For the purpose of this overview, we will focus on $q=4$, the ideas naturally extend to larger even $q$ and we refer to \cref{sec:2_to_q_technical} for the details.
We will focus on proving \cref{thm:main_intro} with the runtime relaxed to $\poly(n,d)$.
The sum-of-squares certificate (\cref{thm:sos_moment_intro}) requires some additional work to show that all steps are captured by the sum-of-squares proof system.
We defer the details to \cref{sec:sos_moment_technical}.

\subsection{Notation and Set-Up}

We use $\gtrsim, \lesssim$ to denote inequalities that are true up to absolute constants.
To avoid dimension-dependent normalization factors, we will make use of \emph{expectation norms}.
For $q \geq 1$ and a vector $x \in \R^n$, we define $\normqexp{x} = \paren{\E_i \abs{x_i}^q}^{1/q}$, where $\E_i$ denotes the uniform distribution over $i \in [n]$.
Note the bar over $q$ in the notation to distinguish it from the regular $q$-norm $\normq{\cdot}$.

In particular, for the rest of this paper (with the exception of \cref{sec:interpolation_pq} on the interpolation to general hypercontractive norms), we will be considering $q$ even and 
\[
    \normtwoqexp{X} = \sup_{\normt{v} = 1} \normqexp{Xv} = \sup_{\normt{v} = 1} \Paren{\E_i \iprod{x_i,v}^q}^{1/q}= n^{-1/q} \sup_{\normt{v} = 1} \normq{Xv} = n^{-1/q} \normtwoq{X} \,.
\]
I.e., $v$ is still unit in the standard Euclidean norm.
Whenever we say `unit' it is with respect to the Euclidean norm.
For a vector $x$, we define $\bar{x} = x / \normt{x}$.
Other norms that we will use are $\normtwoinfty{X}$, i.e., the maximum $\normt{\cdot}$-norm of the rows of $X$, and $\normtwotwoexp{X}$, the square root of the spectral norm of $\tfrac{1}{n} X^\top X = \E_i x_i x_i^\top$.

With these conventions, we say an algorithm gives a $\beta$-approximation for the $2 \rightarrow \bar{4}$ norm if given $X \in \R^{n \times d}$ with rows $x_1, \ldots, x_n \in \R^d$ it can decide whether it satisfies:\footnote{In \cref{sec:2_to_q_technical} we allow a general upper bound of $\alpha^4$ in the \textbf{NO} case.}
\begin{itemize}
    \item \textbf{YES}: $\normtwofourexp{X} > \beta$, that is, there is a unit vector $v$ such that $\E_i \iprod{x_i,v}^4 > \beta^4$.
    \item \textbf{NO}: $\normtwofourexp{X} \leq 1$, that is, for all unit vectors $v$ it holds that $\E_i \iprod{x_i,v}^4 \leq 1$.
\end{itemize}

We say an algorithm that takes $X$ as input \emph{certifies} an upper bound of $\gamma$ on $\normtwofourexp{X}$, when it outputs a number $\ALG(X)$ such that $\normtwofourexp{X} \leq \ALG(X)$ for all $X$ (soundness), independent whether $X$ came from the \textbf{YES} or \textbf{NO} case, and when $X$ came from the \textbf{NO} case, then $\ALG(X) \leq \gamma$ (completeness).
An algorithm certifying an upper bound of $\beta$ directly gives a $\beta$-approximation by computing $\ALG(X)$ and comparing it to $\beta$.
Our algorithms will also be of this form.
Moreover, it will  be easy to extract a vector witnessing large 4-norm in the \textbf{YES} case, i.e., solve the `search' version of the problem.

\subsection{Approximation Algorithms for the \texorpdfstring{$2 \rightarrow 4$}{2-to-4} Norm}

\paragraph{Baseline certificate}

We start by describing the baseline certificate mentioned in \cref{sec:intro}.
In particular, we claim that efficiently certifying an upper bound of $d^{1/4}$ is easy by setting $\ALG(X) = \normop{\E_i x_i^{\otimes 2}[x_i^{\otimes 2}]^\top}^{1/4}$.
This implies an efficient algorithm with approximation factor $d^{1/4}$.
In \cref{sec:baseline} we show that the natural extension to larger $q$ still gives a $d^{1/4}$ approximation.
We also show that this is tight.
I.e., there are matrices $X$ such that $\normtwoqexp{X} \leq O_q(1)$, but $\ALG(X) = \Omega_q(d^{1/4})$.

Returning to $q = 4$, we first argue soundness.
For any $X$, it holds that
\[
    \normtwofourexp{X}^4 = \sup_{\normt{v} = 1}\E_i \iprod{x_i,v}^4 = \sup_{\normt{v} = 1} \Brac{v^{\otimes 2}}^\top \Paren{\E_i x_i^{\otimes 2} \Brac{x_i^{\otimes 2}}^\top} v^{\otimes 2} \leq \Normop{\E_i x_i^{\otimes 2} \Brac{x_i^{\otimes 2}}^\top} = \Paren{\ALG(X)}^4 \,.
\]
To show completeness, consider any $d^2$-dimensional unit vector and reshape it into a $d \times d$ matrix $A$.
We need to show that $\E_i \Paren{x^\top A x}^2\leq d$ assuming $\normtwofourexp{X} \leq 1$, then taking fourth roots implies the claim.
Without loss of generality we can then assume $A$ is symmetric.
Let $\sum_{r=1}^d \lambda_r v_r v_r^\top$ be its Eigen decomposition.
Then it follows by Cauchy-Schwarz that
\[
\E_i \Paren{x_i^\top A x_i}^2 = \E_i \Paren{\sum_{r=1}^d \lambda_r \iprod{x_i,v_r}^2}^2 \leq \Paren{\sum_{r=1}^d \lambda_r^2} \cdot \sum_{r = 1}^d \E_i \iprod{x_i,v_r}^4 = \Normf{A}^2 \cdot \sum_{r = 1}^d \E_i \iprod{x_i,v_r}^4 \,.
\]
But when $\normtwofourexp{X} \leq 1$, $\E_i \iprod{x_i,v_r}^4 \leq 1$ for all $r = 1, \ldots, d$ which implies that the above is at most $d$.

\paragraph{Certificates via a short list of `proxy' directions}

The recent work of Guth, Maldague, and Urschel \cite{guth2025estimating} improves over the baseline certificate.
Their certificate can be cast as checking if $\normqexp{X\hat{v}}$ is small for a short list of \emph{proxy directions} $\hat{v}$ that can be computed efficiently from the input.
The hope is that if the proxy directions are `representative enough', this certifies that $\normfourexp{Xv}$ is not too large for \emph{any} direction $v$.
Our algorithm will also be based on such a list, albeit a more cleverly chosen one.

Formally, given any matrix $X$, they compute in time $\poly(n,d)$ a list of unit vectors $L$ such that
\[
    \normtwofourexp{X} = \sup_{\normt{v} = 1} \normfourexp{Xv} \leq \underbrace{n^{1/12} \cdot \sup_{\hat{v} \in L} \normfourexp{X\hat{v}}}_{\coloneqq \ALG(X)} \leq n^{1/12} \cdot \normtwofourexp{X}\,.
\]
The certification algorithm then outputs $\ALG(X)$.
The inequalities above imply that this certifies an upper bound of $n^{1/12}$.
Further, this algorithm also solves the search problem:
There must exist $\hat{v} \in L$ such that $\normfourexp{X\hat{v}} \geq \normtwofourexp{X}/ n^{1/12}$.
Their list $L$ contains the normalized rows of $X$ and the top right singular vector of $X$.
Recall that typically $n \gtrsim d^2$ (for $q = 4$) and we can always reduce to the case of $n = \tilde{O}(d^2)$ via Lewis-weight sparsification.
Thus, the above gives a $d^{1/6}$-approximation.\footnote{It can be shown that for this concrete choice of list the sparsification pre-processing step is necessary. Without it, the guarantee degrades to $d^{1/4}$ when $n \gtrsim d^3$, see~\cref{thm:guth_limitation}.}

\paragraph{Our starting point: An infinite family of Frobenius norms}
Note that certifying an upper bound on $\normtwofourexp{X}$ is equivalent to certifying an upper bound on the injective norm of the four-tensor $\E_i x_i^{\otimes 4}$.
The baseline certificate is a spectral certificate based on a direct matrix-flattening of this tensor.
An alternative viewpoint on our certificate is to see it as a family of spectral certificates, using more carefully crafted matrices than the one in the baseline.

We start by introducing an infinite family of $d \times d$ matrices that \emph{exactly} capture $\normtwofourexp{X}$.
For a unit vector $v$, define the $d \times d$ matrix $M_v = \E_i \iprod{x_i,v}^2 x_i x_i^\top$, corresponding to contracting only two modes of the four tensor with $v$.
Note that
\[
    \normtwofourexp{X}^4 = \sup_{\normt{v} = 1} \E_i \iprod{x_i,v}^4 = \sup_{\normt{v} = 1} v^\top \Paren{\E_i \iprod{x_i,v}^2 x_i x_i^\top }v \leq \sup_{\normt{v} = 1} \normop{M_v} \,,
\]
where it can be checked using Cauchy-Schwarz that the last inequality is in fact an equality.
Thus, $\normtwofourexp{X}^4 = \sup_{\normt{v} = 1} \normop{M_v}$.

A priori, it is unclear how this might make the problem more tractable.
In particular, $\normop{M_v}$ does not have a nice closed-form expression as a function of $v$.
Our next step is thus to find upper bounds that do have such a nice expression and that are tight enough to certify the $d^{1/8}$ upper bound on $\normtwofourexp{X}$ we are seeking.
We argue that passing to the Frobenius norm of $M_v$ is sufficient.
Indeed, it holds that
\begin{equation*}
    \label{eq:frob_bound}
    \normtwofourexp{X}^4 = \sup_{\normt{v} = 1} \normop{M_v} \leq \sup_{\normt{v} = 1} \normf{M_v} \leq \sqrt{d} \cdot \sup_{\normt{v} = 1} \normop{M_v} = \sqrt{d} \cdot  \normtwofourexp{X}^4\,.%
\end{equation*}
Thus, when $\normtwofourexp{X} \leq 1$, it follows that $\sup_{\normt{v} = 1} \normf{M_v} \leq \sqrt{d}$. 
In the remainder of this section our main goal is to \emph{certify} this upper bound.
Given this, the above inequality implies that this also certifies an upper bound of $d^{1/8}$ on $\normtwofourexp{X}$.

\paragraph{From spectral norms to proxy directions}

We will be using the idea of proxy directions discussed above.
We take a two-step approach.
First, we will carefully construct a $d \times d$ matrix $\tilde{M}$, that can be efficiently computed from $X$, such that
\begin{equation}
    \label{eq:tilde_M_cert}
\hphantom{\textnormal{(Spectral Bound)}}\normtwofourexp{X} \leq \sup_{\normt{v} = 1} \normf{M_v}^{1/4} \leq \normop{\tilde{M}}^{1/8} \,.
\tag{Spectral Bound}
\end{equation}
Then, we will construct a polynomially-sized list of vectors $L$ such that
\begin{equation}
    \label{eq:cert_list}
    \hphantom{\textnormal{(List Bound)}}\normop{\tilde{M}}^{1/8} \leq d^{1/8} \cdot \max_{\hat{v} \in L} \normfourexp{X \hat{v}} \leq d^{1/8} \cdot \normtwofourexp{X}\,.
    \tag{List Bound}
\end{equation}
Our certificate is $\ALG(X) = d^{1/8} \cdot \max_{\hat{v} \in L} \normfourexp{X \hat{v}}$.
Soundness and completeness follow from  \cref{eq:cert_list}.
Further, we can efficiently compute a unit vector $\hat{v}$ in the list $L$ such that $\normfourexp{X \hat{v}} \geq \normtwofourexp{X} / d^{1/8}$.
Concretely, the list $L$ we construct contains the following unit vectors in $\R^d$: %
\begin{enumerate}
    \item The standard basis vectors $e_1, \ldots, e_d$.
    \item The normalized rows of $X$: $\bar{x}_1 = \frac{x_1}{\normt{x_1}}, \ldots, \bar{x}_n= \frac{x_n}{\normt{x_n}}$.
    \item The top eigenvector of $\tilde{M}$.
    \item The top eigenvectors of the matrices $M_1, \ldots, M_n$ where $M_i = \E_j \iprod{x_i,x_j}^2 x_j x_j^\top$. 
\end{enumerate}
The meaning of the $M_i$ matrices will become clear in the proof.

\paragraph{Proving the Spectral Bound}

The difficult part of showing \cref{eq:tilde_M_cert} is to construct the right matrices $\tilde{M}$ and $M_i$.
We set $\tilde{M} = \E_i \normop{M_i} x_i x_i^\top$ and $M_i$ as above.
For some intuition on how to arrive at this choice, we can expand $\normf{M_v}^2$ and try to reduce it to a sequence of spectral certificates.
Given this setup, we can verify \cref{eq:tilde_M_cert}.
We have already shown the first inequality.
For the second one, we raise to the eighth power and note that for any unit vector $v$ it holds that 
\begin{align*}
    \normf{M_v}^2 &= \E_{ij} \iprod{x_i,v}^2 \iprod{x_j,v}^2 \iprod{x_i,x_j}^2 = \E_i \iprod{x_i,v}^2 v^\top \Paren{\E_j \iprod{x_i,x_j}^2 x_j x_j^\top}v \leq \E_i \iprod{x_i,v}^2 \normop{M_i} \\
    &= v^\top \Paren{\E_i \normop{M_i} x_i x_i^\top} v \leq \normop{\tilde{M}} \,.
\end{align*}
Thus, $\sup_{\normt{v} = 1} \normf{M_v}^{1/4} \leq \normop{\tilde{M}}^{1/8}$.

Recall that $\normtwofourexp{X} \leq 1$ implies that $\normf{M_v}^2 \leq d$.
To gain intuition why the above should yield a tight upper bound, we can consider the idealized setting in which the rows are isotropic, i.e., $\E_i x_i x_i^\top = I_d$ and all pair-wise inner products are of order $\sqrt{d}$ (we show that this is still true on average when $\normtwofourexp{X} \leq 1$).
Then, ignoring dependencies between the rows we obtain $M_i \approx d \cdot \E_i x_i x_i^\top =d \cdot I_d$ and $\tilde{M} \approx d \cdot \E_i x_i x_i^\top = d \cdot I_d$, which gives $\normop{\tilde{M}} \lesssim d$.

\paragraph{Proving the List Bound}
Let $B = \max_{\hat{v} \in L} \normfourexp{X \hat{v}}$, we will show that $\normop{\tilde{M}} \leq d \cdot B^8$ which implies \cref{eq:cert_list}.
At the end of this section, we will argue that for all $i$ it holds that $\normop{M_i} \leq B^4 \normt{x_i}^2$.
Admitting this for the moment, this implies that
\[
    \normop{\tilde{M}} = \normop{\E_i \normop{M_i} x_i x_i^\top} \leq B^4 \cdot \normop{\E_i \normt{x_i}^2 x_i x_i^\top} \,.
\]

We will next argue that bounded $2 \rightarrow \bar{4}$ norm implies that row norms are small on average.
Indeed, we will show that $\E_i \normt{x_i}^4 \leq B^4 d^2$.
Denote by $x_i(r)$ the $r$-th coordinate of $x_i$, then by Cauchy-Schwarz
\[
     \E_i \normt{x_i}^4 = \sum_{r,r' = 1}^d \E_i x_i(r)^2 x_i(r')^2 \leq \sum_{r,r' = 1}^d \sqrt{\E_i x_i(r)^4} \sqrt{\E_i x_i(r')^4} = \Paren{\sum_{r=1}^d \sqrt{\E_i x_i(r)^4} }^2\,.
\]
But since $e_r \in L$ for all $r$, we know that $\sqrt{\E_i x_i(r)^4} =  \normfourexp{X e_r}^2 \leq B^2$, which gives the claim.\footnote{A different approach is as follows: After sparsifying down to $n = \tilde{O}(d^2)$, it holds for \emph{all} $i$ that $\normt{x_i}^4 = \iprod{x_i,\bar{x}_i}^4 = n \cdot \normfourexp{X \bar{x}_i}^4 \leq \tilde{O}(d^2) \cdot B^4$. As a consequence, the coordinate vectors $e_i$ could be eliminated from the list. We choose to present the argument above since it automatically handles instances when $n \gg d^2$, without sparsifying first.}

With this in hand, we turn back to analyzing $\tilde{M}$.
Let $\tilde{u} \in L$ be the top eigenvector of $\tilde{M}$, then 
\[
    \normop{\tilde{M}} \leq B^4 \cdot \E_i \normt{x_i}^2 \iprod{x_i, \tilde{u}}^2 \leq B^4 \cdot \sqrt{\E_i \normt{x_i}^4} \cdot \sqrt{\E_i \iprod{x_i, \tilde{u}}^4} \leq B^4 \cdot  \sqrt{B^4 d^2} \cdot \normfourexp{X \tilde{u}}^2 \leq B^8 d\,.
\]

It remains to show that $\normop{M_i} \leq B^4\normt{x_i}^2$.
Let $\tilde{u}_i \in L$ be the top eigenvector of $M_i$.
Then
\[
    \normop{M_i} = \E_j \iprod{x_i,x_j}^2\iprod{x_j,\tilde{u}_i}^2 \leq \sqrt{\E_j \iprod{x_i,x_j}^4} \sqrt{\E_j \iprod{x_j,\tilde{u}_i}^4} \,.
\]
The second factor is at most $B^2$.
The first term is equal to $\normfourexp{X x_i}^2 = \norm{x_i}^2 \cdot \normfourexp{X \bar{x}_i}^2 \leq \norm{x_i}^2 \cdot B^2$.

\subsection{Open Questions}
\label{sec:open_questions}

We end with a few open questions and future directions.

\paragraph{Better approximation factors for the $2 \rightarrow q$ norm}
The first concerns even better approximation factors for the $2 \rightarrow q$ norm problem in (quasi-)polynomial time.
In constructing our certificates we used the bound $\normtwofourexp{X}^4 = \sup_{\normt{v} = 1} \normop{M_v} \leq \sup_{\normt{v} = 1} \normf{M_v}$.
The right-hand side can be as large as $\sqrt{d} \normtwofourexp{X}^4$, showing that the $d^{1/8}$ approximation we gave is best-possible using this approach.
A natural avenue to obtaining better approximation factors would be to rely on higher-order Schatten-$p$ norms and instead use the bound $\sup_{\normt{v} = 1} \normop{M_v} \leq \sup_{\normt{v} = 1} \norm{M_v}_p$, for $p > 2$.
Indeed, it immediately follows that $\sup_{\normt{v} = 1} \norm{M_v}_p \leq d^{1/p} \cdot \normtwofourexp{X}^4$.
Thus, if we can also \emph{certify} this upper bound of $d^{1/p}$, this would lead to an $d^{1/(4p)}$ approximation for the $2 \rightarrow 4$ norm.
Taking $p \approx \log d$ would then correspond to a constant factor approximation. 
The analysis of $p = 2$ in this work does not seem to extend to the more general case.

\paragraph{Extension to Best Separable State}
The second concerns the best separable state (BSS) problem.
While existing reductions \cite{harrow2013testing,barak2012hypercontractivity} are too lossy to translate the guarantees of \cref{thm:main_intro} to BSS, it would be intriguing to see if the same ideas can lead to improved approximation algorithms in this setting too.
The main difficulty seems to be that in the above we crucially exploited that the tensor $\E_i x_i^{\otimes 4}$ is fully symmetric.

\section{Approximating \texorpdfstring{ $2\rightarrow q$ Norms}{2-to-q Norms}}
\label{sec:2_to_q_technical}
In this section, we will prove the full version of \cref{thm:main_intro}.
Compared to \cref{sec:tech_overview} we also allow constants different from one in the \textbf{NO} case.
We say an algorithm gives an $(\alpha, \beta)$-approximation for the $2 \rightarrow q$ norm if it solves the following:
Given an input $X \in \R^{n \times d}$ with rows $x_1, \ldots, x_n$ satisfying one of the following, decide which one it satisfies.
\begin{itemize}
    \item \textbf{YES}: $\normtwoqexp{X} > \beta$, that is, there is a unit vector $v$ such that $\E_i \iprod{x_i,v}^q > \beta^q$.
    \item \textbf{NO}: $\normtwoqexp{X} \leq \alpha$, that is, for all unit vectors $v$ it holds that $\E_i \iprod{x_i,v}^q \leq \alpha^q$.
\end{itemize}

The following theorem gives an $(\alpha, \beta)$-approximation to the $2 \rightarrow q$ norm problem as long as $\beta / \alpha \geq d^{1/4 - 1/(2q)}$.
Setting $\alpha = 1$ recovers \cref{thm:main_intro}.
\begin{theorem}
    \label{thm:two_to_q_full}
    Let $X \in \R^{n \times d}$, $q \geq 4$ be even, and $\alpha > 0$.
    There exists an algorithm running in time $O(n^2 d^2 + nd^3)$ that given $X$, outputs a list $L$ containing $O(n + d)$ unit vectors such that
    \[
        \normtwoqexp{X} \leq d^{1/4 - 1/(2q)} \cdot \max_{\hat{v} \in L} \normqexp{X \hat{v}} \leq d^{1/4 - 1/(2q)} \cdot  \normtwoqexp{X} \,.
    \] 
    Consequently, by computing $\normqexp{X \hat{v}}$ for all $\hat{v} \in L$, it follows that in the same time we can
    \begin{enumerate}
        \item Certify that $\normtwoqexp{X} \leq \alpha \cdot d^{1/4- 1/(2q)}$, if $\normtwoqexp{X} \leq \alpha$.
        \item Compute a unit vector $\hat{v}$ such that $\normqexp{X \hat{v}} \geq \normtwoqexp{X}/d^{1/4 - 1/(2q)}$. 
    \end{enumerate}
\end{theorem}

Thus, for the case of $q = 4$ this gives a $d^{1/8}$-approximation.
The proof exactly follows the outline in \cref{sec:tech_overview}.
The only thing that changes is that now $M_i = \E_j \iprod{x_i,x_j}^{q-2}x_jx_j^\top$, we still use $\tilde{M} = \E_i \normop{M_i} x_i x_i^\top$.

\begin{proof}[Proof of \cref{thm:two_to_q_full}]

    Let $q$ be an even integer.
    The proof of the "Consequently, " part is straightforward by checking $\normqexp{X \hat{v}}$ for all $\hat{v} \in L$:
    Since all $\hat{v}$ are unit, $\normtwoqexp{X} \leq \alpha$ implies that $\max_{\hat{v} \in L} \normqexp{X\hat{v}} \leq \alpha$.
    Thus, $\normtwoqexp{X} \leq \alpha \cdot d^{1/4 - 1/(2q)}$.
    The second part follows by rearranging.

    For the rest of the proof we assume without loss of generality that none of the rows of $X$ are 0.

    \paragraph{Reducing to the operator norm of $\tilde{M}$}
    Let $M_i \coloneqq \E_j \iprod{x_i,x_j}^{q-2} x_j x_j^\top$ and consider the matrix $$\tilde{M} = \E_i \normop{M_i} x_i x_i^\top \,.$$
    Then for any unit vector $v$ and $M_v = \E_i \iprod{x_i,v}^2 x_i^{\otimes (q-2)}$, it holds that
    \[
        \normqexp{Xv}^q = \E_i \iprod{x_i,v}^q = \Iprod{v^{\otimes(q-2)}, M_v} \leq \Normf{M_v} = \sqrt{\E_{ij} \iprod{x_i,v}^2 \iprod{x_j,v}^2 \iprod{x_i,x_j}^{q-2}} \,.
    \]
    Next, notice that 
    \begin{align*}
        \Normf{M_v}^2 = \E_{ij} \iprod{x_i,v}^2 \iprod{x_j,v}^2 \iprod{x_i,x_j}^{q-2} &= \E_i \iprod{x_i,v}^2 \cdot v^\top \Paren{\E_j \iprod{x_i,x_j}^{q-2} x_j x_j^\top}v= \E_i \iprod{x_i,v}^2 \cdot v^\top M_i v \\
        &\leq \E_i \iprod{x_i,v}^2 \cdot \normop{M_i} = v^\top \Paren{\E_i \normop{M_i} \cdot x_i x_i^\top}v =v^\top \tilde{M}v \\
        &\leq \normop{\tilde{M}} \,.
    \end{align*}
    Taking a supremum over all unit vectors $v$ then gives $\normtwoqexp{X}^{2q} \leq \normop{\tilde{M}}$.

    \paragraph{Constructing $L$ and upper bounding $\normop{\tilde{M}}$}

    Let $e_i$ be the standard basis vectors in $\R^d$, $\bar{x}_i$ the normalized rows of $X$, $u_i$ the top eigenvectors of the $M_i$ matrices, and $u$ the top eigenvector of $\tilde{M}$.
    We set
    \[
        L = \Set{e_1, \ldots, e_d, \bar{x}_1, \ldots, \bar{x}_n, u_1, \ldots, u_n, u} \,.
    \]
    Clearly $L$ can be computed in time $\poly(n,d)$ given $X$, we analyze the precise runtime below.
    Further, let $B = \max_{\tilde{v} \in L} \normqexp{X \tilde{v}}$.

    We first derive an upper bound for $\normop{M_i}$.
    Note that $\normqexp{X u_i} \leq B$ and $\normqexp{X x_i} \leq \normt{x_i} \cdot B$ for all $i \in [n]$.
    By H\"older's Inequality it follows that
    \begin{equation}
    \begin{aligned}
        \label{eq:M_i_spectral_bound}
        \normop{M_i} &= u_i^\top M_i u_i = \E_j \iprod{x_i,x_j}^{q-2} \iprod{x_j,u_i}^2 \leq \Paren{\E_j \iprod{x_i,x_j}^q }^{1-2/q} \cdot \Paren{\E_j \iprod{x_j,u_i}^q}^{2/q} \\
        &= \normqexp{X x_i}^{q-2} \cdot \normqexp{X u_i}^2 \leq B^q \cdot \normt{x_i}^{q-2} \,.
    \end{aligned}
    \end{equation}

    Thus, averaging over all $i$ and using H\"older's Inequality again, we obtain that
    \begin{equation}
    \begin{aligned}
        \label{eq:M_tilde_spectral_bound}
        \normop{\tilde{M}} &= u^\top \tilde{M}u = \E_i \normop{M_i} \cdot \iprod{x_i,u}^2 \leq B^q \cdot \E_i \norm{x_i}^{q-2} \cdot \iprod{x_i,u}^2 \\
        &\leq B^q \cdot \Paren{\E_i \norm{x_i}^q}^{1-2/q} \cdot \Paren{\E_i \iprod{x_i, u}^q}^{2/q}  = B^q \cdot \Paren{\E_i \norm{x_i}^q}^{1-2/q} \cdot \normqexp{Xu}^2 \\
        &\leq B^{q+2} \cdot \Paren{\E_i \norm{x_i}^q}^{1-2/q} \,.
    \end{aligned}
    \end{equation}

    It remains to bound $\E_i \norm{x_i}^q$.
    Let $\gamma$ denote ordered tuples in $[d]^{q/2}$.
    It follows by a standard calculation using the AM-GM inequality that $\E_i \norm{x_i}^q \leq B^q \cdot d^{q/2}$.
    Indeed,
    \begin{equation}
    \begin{aligned}
        \E_j \norm{x_j}^q &= \E_j \Paren{\sum_{r=1}^d x_j(r)^2}^{q/2} = \sum_{\gamma} \E_j x_j^{2\gamma} = \sum_{\gamma} \E_j \prod_{r = 1}^{q/2}x_j^2(\gamma_r) \leq \sum_{\gamma} \E_j \tfrac 2 q \sum_{r = 1}^{q/2}x_j^q(\gamma_r) \\
        &= \sum_{\gamma} \tfrac 2 q \sum_{r = 1}^{q/2} \E_j\iprod{x_j, e_{\gamma_r}}^q  = \sum_{\gamma} \tfrac 2 q \sum_{r = 1}^{q/2} \normqexp{X e_{\gamma_r}}^q \leq B^q \cdot \sum_{\gamma}  1 = B^q \cdot d^{q/2} \,.
    \end{aligned}
    \end{equation}

    Thus, overall we obtain that $\normop{\tilde{M}} \leq B^{2q} \cdot d^{(q-2)/2}$, which in turn implies that $\normtwoqexp{X} \leq \normop{\tilde{M}}^{1/2q} \leq B \cdot d^{\tfrac 1 4 - \tfrac 1 {2q}}$.

    \paragraph{Runtime}

    Checking all directions $e_j$ and $\bar{x}_i$ can be done in time $O(n^2 d)$.
    It takes time $O(nd^2)$ to construct a single matrix $M_i$ and additionally time $O(d^3)$ to compute its spectral norm and top eigenvector.
    Thus, checking all $u_i$ takes time $O(n^2d^2 + nd^3 + nd) = O(n^2 d^2 + nd^3)$.
    Lastly, to compute $\tilde{M}$ and its top eigenvector then takes time $O(nd^2 + d^3 + nd^2 +d^3) = O(nd^2 + d^3)$.
    Checking this top eigenvector again takes time $O(nd)$.
    The total runtime is thus $O(n^2d^2 + nd^3)$.

\end{proof}

\subsection{Sum-of-Squares Certificate}
\label{sec:sos_moment_technical}

In this section we show that there also exists a sum-of-squares certificate for the upper bound in \cref{thm:two_to_q_full}.
We first recall the definition of (certifiably and non-certifiably) bounded and hyper-contractive distributions from \cite{kothari2018robust,hopkins2018mixture}.

Let $\cD$ be a distribution over $\R^d$ with mean zero and covariance $\Sigma$.
We say that $\cD$ is 
\begin{itemize}
    \item (non-certifiably) $(B_q,q)$ bounded if for any $v \in \R^d$, $\E_{x \sim \cD} \iprod{x,v}^q \leq B_q^q \normt{v}^q$.
    \item (non-certifiably) $(B_q,q)$ hyper-contractive if for any $v \in \R^d$, $\E_{x \sim \cD} \iprod{x,v}^q \leq B_q^q \iprod{v, \Sigma v}^{q/2}$.
\end{itemize}

For the certifiable versions, we say a polynomial $p$ in formal variables $v$ is a sum of squares of degree $\ell$ if it can be written as $p(v) = \sum_r q_r(v)^2$ and each $q_r^2$ is a polynomial of degree at most $\ell$.
We then define the following \cite{kothari2018robust,hopkins2018mixture}.
\begin{definition}
    \label{def:cert_moments}
    Let $q \in \N$ be even and $B_q > 0$. We say a distribution over $\R^d$ with covariance $\Sigma$ is
    \begin{enumerate}
        \item \emph{$(B_q,q)$-certifiably bounded at degree $\ell$}, if $p(v) = B_q^q \cdot \normt{v}^q - \E_{x \sim \cD} \iprod{x,v}^q$ is a sum of squares of degree $\ell$.
        \item \emph{$(B_q,q)$-certifiably hyper-contractive at degree $\ell$} if $p'(v) = B_q^q \cdot \iprod{v,\Sigma v}^{q/2} - \E_{x \sim \cD} \iprod{x,v}^q$ is a sum of squares of degree $\ell$.
    \end{enumerate}
\end{definition}

We show that any distribution that is (non-certifiably) bounded/hyper-contractive is also certifiably bounded/hyper-contractive with a $d^{1/4 - 1/(2q)}$ loss in parameters.
\stefan{Probably the degree is something like $4q$ or so}
\begin{theorem}[Full version of \cref{thm:sos_moment_intro}]
    \label{thm:sos_moment_technical}
    Let $q \in \N$ be even and $\cD$ be a distribution over $\R^d$.
    Then,
    \begin{enumerate}
        \item If $\cD$ is $(B_q,q)$ bounded, then $\cD$ is also $(d^{1/4 - 1/(2q)} B_q,q)$ certifiably bounded at degree $O(q)$.
        \item If $\cD$ is $(B_q,q)$ hyper-contractive, then $\cD$ is also $(d^{1/4 - 1/(2q)} B_q,q)$ certifiably hyper-contractive at degree $O(q)$.
    \end{enumerate}
\end{theorem}

To show this theorem, we show that the proof of \cref{thm:two_to_q_full} is captured by the degree-$O(q)$ sum-of-squares proof system, defined below.

\paragraph{Sum-of-squares basics}

To keep the exposition self-contained, we first recall some basics about the sum-of-squares proof system.
Let $\cA = \set{r_1(X), \ldots, r_m(X) \geq 0}$ be a (possibly empty) system of polynomial inequalities in formal variables $X$.
Let $p$ be another polynomial in $X$, we say that $\cA$ implies $p$ at degree $t$, denoted by $\cA \proves{t}{X} p \geq 0$, if we can write $p(X) = \sum_{S \sse [m]} b_S(X) \prod_{i \in S} r_i(X)$ where each $b_S$ is a sum of squares and each summand has degree at most $t$.
When $\cB = \set{r'_1(X) \geq 0, \ldots, r'_{m'}(X) \geq 0}$ is another system of polynomial inequalities, we write $\cA \proves{t}{X} \cB$ if $\cA \proves{t}{X} r'_j \geq 0$ for all $j \in [m']$.
We use the following transitivity rule: If $\cA \proves{t}{X} \cB$ and $\cB \proves{t'}{X} p \geq 0$, then $\cA \proves{t \cdot t'}{X} p \geq 0$.
We write $\cA \proves{t}{X} p \leq p'$ if $\cA \proves{t}{X} p' - p \geq 0$.

We will use the following facts.
The first one follows directly since there is a matrix $L \in \R^{d \times d}$ such that $A = \normop{A} \cdot I_d - LL^\top$, the second one we prove at the end of this section.
\begin{fact}
    \label{fact:spectral_sos}
    Let $A \in \R^{d \times d}$ be a real-valued symmetric matrix and $v$ be a $d$-dimensional indeterminate.
    Then $\proves{2}{v} v^\top A v \leq \normop{A} \normt{v}^2$ \,.
\end{fact}
\begin{fact}
    \label{fact:matrix_cs_sos}
    Let $B > 0$ and $q \geq 4$ be even.
    Let $v$ be a $d$-dimensional vector-valued indeterminate, and $M$ be a $d^{q-2}$-dimensional tensor-valued indeterminate.
    Then
    \[
        \Set{\normf{M}^2 \leq B^2  \normt{v}^4} \proves{2q}{M,v} \iprod{v^{\otimes(q-2)}, M } \leq B  \normt{v}^q\,.
    \]
\end{fact}

\paragraph{Proof of \cref{thm:sos_moment_technical}}
Recall that in the proof of \cref{thm:two_to_q_full}, we started with $\E_x \iprod{x,v}^q \leq \normf{M_v}$ for $M_v = \E_x \iprod{x,v}^2 x^{\otimes(q-2)}$
The only technical complication is that the Frobenius norm of $M_v$ is not a polynomial in the formal variables $v$.
However, we can still obtain a bound on the square of this Frobenius norm and conclude with \cref{fact:matrix_cs_sos}.

\begin{proof}
    The second part (certifiable hyper-contractivity) can be derived from the first part using a standard whitening argument (via the pseudo-inverse of $\Sigma$).

    Whenever we omit the degree of the sum-of-squares proof, we will implicitly assume that it is an absolute constant, independent of $q$.
    Let $M_v = \E_{x \sim \cD} \iprod{x,v}^2 x^{\otimes (q-2)}$, we will show that $\proves{}{v} \normf{M_v}^2 \leq (B_q^q \cdot d^{q/4 - 1/2})^2 \cdot \normt{v}^4$ shortly.
    Assuming this, we obtain using transitivity and Fact~\ref{fact:matrix_cs_sos} that
    \[
        \proves{O(q)}{v} \E_{x \sim \cD} \iprod{x,v}^q = \iprod{v^{\otimes(q-2)}, M_v} \leq B_q^q \cdot d^{q/4 - 1/2} \cdot \normt{v}^q
    \]
    which gives the conclusion.

    To derive the SoS upper bound on $\normf{M_v}^2$, let $\tilde{M} = \E_{x \sim \cD} \normop{M_x} xx^\top$, where $M_x = \E_{x' \sim \cD}\iprod{x,x'}^{q-2} x' {x'}^\top$ for a fixed vector $x$.
    Using Fact~\ref{fact:spectral_sos} twice it follows that 
    \begin{align*}
        \proves{}{v,M_v} \normf{M_v}^2 &= \E_{x,x' \sim \cD} \iprod{x,v}^2 \iprod{x',v}^2 \iprod{x,x'}^{q-2} = \E_{x \sim \cD} \iprod{x,v}^2 \cdot  v^\top M_x v \leq \Paren{\E_{x \sim \cD} \iprod{x,v}^2 \cdot \normop{M_x}} \cdot \normt{v}^2 \\
        &= v^\top \tilde{M} v \cdot \normt{v}^2 \leq \normop{\tilde{M}} \cdot \normt{v}^4 \,.
    \end{align*}
    Thus, to complete the proof it suffices to show the (non-SoS) bound $\normop{\tilde{M}} \leq (B_q^q \cdot d^{q/4 - 1/2})^2$, this will be basically the same as in the proof of \cref{thm:two_to_q_full}.

    \paragraph{Bounding the spectral norm of $\tilde{M}$}

    We first claim that $\normop{M_x} \leq B_q^q \normt{x}^{q-2}$ for any $x$.
    Indeed, let $u$ be the top eigenvector of $M_x$.
    Note that this is a fixed real vector, \emph{not} a sum-of-squares variable.
    By H\"older's Inequality and using that $\E_{x' \sim \cD} \iprod{x,x'}^q \leq B_q^q \normt{x}^q$ we obtain
    \[
        \normop{M_x} = \E_{x'} \iprod{x,x'}^{q-2} \iprod{x',u}^2 \leq \Paren{\E_{x'} \iprod{x,x'}^q}^{(q-2)/q} \cdot \Paren{\E_{x'} \iprod{x',u}^q}^{2/q} \leq B_q^q \cdot \normt{x}^{q-2} \,.
    \]

    Next, let $\tilde{u}$ be the top eigenvector of $\tilde{M}$.
    We remark that exactly the same argument as in the proof of \cref{thm:two_to_q_full} implies that $\E_x \normt{x}^q \leq B_q^q d^{q/2}$.
    Then again by H\"older's Inequality and the previous bound we obtain 
    \[
        \normop{\tilde{M}} = \E_x \normop{M_x} \iprod{x,\tilde{u}}^2 \leq B_q^q \cdot \E_x \normt{x}^{q-2} \iprod{x,\tilde{u}}^2 \leq B_q^q \cdot \Paren{\E_x \normt{x}^q}^{(q-2)/q} \Paren{\E_x \iprod{x,\tilde{u}}^q}^{2/q} \leq B_q^{2q} \cdot d^{q/2 -1} \,,
    \]
    which is equal to $(B_q^q \cdot d^{q/4 - 1/2})^2$.
\end{proof}

\paragraph{Proof of Fact~\ref{fact:matrix_cs_sos}}

The proof uses the following basic SoS facts.
The first one follows since $(a-b)^2 \geq 0$ for any $a,b$.
The first fact implies that $\proves{2}{A,B} \iprod{A,B} \leq \tfrac 1 2 \normt{A}^2 + \tfrac 1 2 \normt{B}^2$, applying this to $A = X \otimes Y$ and $B = Y \otimes X$ gives the second fact.
\begin{enumerate}
    \item For any $\gamma > 0$, $\proves{}{X,Y} X Y \leq \tfrac{1}{2\gamma} X^2 + \tfrac{\gamma}{2} Y^2$ (AM-GM).
    \item $\proves{4}{X,Y} \iprod{X,Y}^2 \leq \normt{X}^2 \normt{Y}^2$ (Cauchy-Schwarz).
\end{enumerate}

\begin{proof}
Let $\alpha, \beta$ be ordered tuples over $[d]^{q/2}$ and $[d]^{(q-4)/2}$ respectively.
For a fixed $\alpha$, denote by $M_\alpha$ the $d^{(q-4)/2}$-dimensional collection of variables that contains all variables $M_{\alpha,\beta}$ for all possible $\beta$.
Let $\gamma > 0$, it then follows by first using AM-GM and then Cauchy-Schwarz that
\begin{align*}
    \Set{\normf{M}^2 \leq B^2 \normt{v}^4}&\proves{}{v,M} \iprod{v^{\otimes(q-2)}, M} = \iprod{v^{\otimes q/2} \otimes v^{\otimes(q-4)/2}, M} = \sum_{\alpha \in [d]^{q/2}} v^\alpha \iprod{v^{\otimes(q-4)/2}, M_\alpha} \\
    &\leq \frac{1}{2\gamma } \sum_{\alpha \in [d]^{q/2}} v^{2\alpha} + \frac{\gamma}{2} \sum_{\alpha \in [d]^{q/2}}\iprod{v^{\otimes(q-4)/2}, M_\alpha} ^2 \\
    &\leq \frac{1}{2\gamma } \sum_{\alpha \in [d]^{q/2}} v^{2\alpha} + \frac{\gamma}{2} \sum_{\alpha \in [d]^{q/2}} \normt{v^{\otimes(q-4)/2}}^2 \cdot \normf{M_\alpha}^2 \\
    &= \frac{1}{2\gamma } \normt{v}^q + \frac{\gamma}{2} \normt{v}^{q-4} \cdot \normf{M}^2 \leq \frac{1}{2\gamma } \normt{v}^q + \frac{\gamma B^2}{2} \normt{v}^{q} \,.
\end{align*}
Choosing $\gamma = 1/B$ then finishes the proof.
\end{proof}

\subsection{Interpolation to Other Hypercontractive Norms}
\label{sec:interpolation_pq}
Using standard interpolation arguments, we can extend \cref{thm:two_to_q_full} to other $p \rightarrow q$ norms in the hypercontractive setting, for arbitrary $1 \leq p < q$ and even $q$.
To avoid confusion, in this section we always use standard $\ell_p$ and $\ell_q$ norms instead of expectation norms.
\begin{theorem}
    \label{thm:p_to_q_full}
    Let $X \in \R^{n \times d}$ and $\alpha> 0$.
    Let $q \in \N$ be even and $1 \leq p \leq q$.
    Let $\gamma_p = 1/p - 1/2$ when $p \leq 2$ and $\gamma_p = 1/2 - 1/p$ when $p \geq 2$.
    There exists an algorithm running in time $O(n^2 d^2 + nd^3)$ that given $X$, outputs a list $L$ containing $O(n + d)$ vectors such that
    \[
        \normpq{X} \leq d^{\gamma_p + 1/4 - 1/(2q)} \cdot \max_{\hat{v} \in L} \frac{\normq{X \hat{v}}}{\norm{\hat{v}}_p} \leq d^{\gamma_p + 1/4 - 1/(2q)} \cdot  \normpq{X} \,.
    \] 
    Consequently, by computing $\normq{X \hat{v}}$ for all $\hat{v} \in L$, it follows that in the same time we can
    \begin{enumerate}
        \item Certify that $\normpq{X} \leq \alpha \cdot d^{\gamma_p + 1/4 - 1/(2q)}$, if $\normpq{X} \leq \alpha$.
        \item Compute a vector $\hat{v}$ such that $\norm{\hat{v}}_p = 1$ and $\normq{X \hat{v}} > \normpq{X}/d^{\gamma_p + 1/4 - 1/(2q)}$.
    \end{enumerate}
\end{theorem}

In particular, this follows immediately by combining \cref{thm:two_to_q_full} with the following (basic) interpolation lemma, proved for completeness at the end of this section.
\begin{fact}
    \label{fact:pq_interpolation}
    For any $X \in \R^{n \times d}$ and $1 \leq p \leq q$ it holds that
    \begin{enumerate}
        \item If $p \leq 2$, then $\normtwoq{X} \leq d^{1/p - 1/2}\normpq{X} \leq d^{1/p - 1/2} \normtwoq{X}$.
        \item If $p \geq 2$, then $\normtwoq{X} \leq \normpq{X} \leq d^{1/2 - 1/p} \normtwoq{X}$.
    \end{enumerate}
\end{fact}
We will also use the following standard estimates to compare the $\ell_p$ and $\ell_2$ norm.

\begin{fact}
    \label{fact:pvsq}
    Let $x \in \R^d$ and $p \geq 1$.
    \begin{enumerate}
        \item If $p \leq 2$, then $d^{1/2 - 1/p}\norm{x}_p \leq  \normt{x} \leq \norm{x}_p$.
        \item If $p \geq 2$, then $\norm{x}_p \leq  \normt{x} \leq d^{1/2 - 1/p}\norm{x}_p$.
    \end{enumerate}
\end{fact}

\begin{proof}[Proof of \cref{thm:p_to_q_full}]
    We split the proof into two cases, depending on whether $p$ is at most or at least 2.
    The two parts are almost identical, except that the inequalities in Fact~\ref{fact:pq_interpolation} are slightly different.
    The algorithm is simple: It invokes the algorithm in \cref{thm:two_to_q_full} and forwards the list $L$ it outputs.
    We only show that this list satisfies 
    \[
        \normpq{X} \leq d^{\gamma_p + 1/4 - 1/(2q)} \cdot \max_{\hat{v} \in L} \frac{\normq{X \hat{v}}}{\norm{\hat{v}}_p} \leq d^{\gamma_p + 1/4 - 1/(2q)} \cdot  \normpq{X} \,.
    \]
    The last part of the theorem is identical to \cref{thm:two_to_q_full}.
    We know that $L$ satisfies\footnote{Here we multiply the guarantee in \cref{thm:two_to_q_full} by $n^{1/q}$ to switch from expectation norms to regular norms.}
    \[
        \normtwoq{X} \leq d^{1/4 - 1/(2q)} \cdot \max_{\hat{v} \in L} \frac{\normq{X \hat{v}}}{\normt{\hat{v}}} \leq d^{1/4 - 1/(2q)} \cdot  \normtwoq{X} 
    \]

    We start with $p \leq 2$.
    In this case $\gamma_p = 1/p - 1/2$
    Using Facts~\ref{fact:pq_interpolation} and \ref{fact:pvsq} we obtain
    \begin{align*}
        \normpq{X} &\leq \normtwoq{X} \leq d^{1/4 - 1/(2q)} \cdot \max_{\hat{v} \in L} \frac{\normq{X \hat{v}}}{\normt{\hat{v}}} = d^{1/4 - 1/(2q)} \cdot \max_{\hat{v} \in L} \frac{\normq{X \hat{v}}}{\norm{\hat{v}}_p} \frac{\norm{\hat{v}}_p}{\normt{\hat{v}}} \\
        &\leq d^{\gamma_p} d^{1/4 - 1/(2q)} \cdot \max_{\hat{v} \in L} \frac{\normq{X \hat{v}}}{\norm{\hat{v}}_p} \leq d^{\gamma_p + 1/4 - 1/(2q)} \cdot \normpq{X} \,.
    \end{align*}
     
    The proof of $p \geq 2$ is completely analogous.
    In this case $\gamma_p = 1/2 - 1/p$
    \begin{align*}
        \normpq{X} &\leq d^{\gamma_p}\normtwoq{X} \leq d^{\gamma_p + 1/4 - 1/(2q)} \cdot \max_{\hat{v} \in L} \frac{\normq{X \hat{v}}}{\normt{\hat{v}}} = d^{\gamma_p + 1/4 - 1/(2q)} \cdot \max_{\hat{v}} \frac{\normq{X \hat{v}}}{\norm{\hat{v}}_p} \frac{\norm{\hat{v}}_p}{\normt{\hat{v}}} \\
        &\leq d^{\gamma_p + 1/4 - 1/(2q)} \cdot \max_{\hat{v} \in L} \frac{\normq{X \hat{v}}}{\norm{\hat{v}}_p} \leq d^{\gamma_p + 1/4 - 1/(2q)} \cdot \normpq{X} \,.
    \end{align*}
    
\end{proof}

\paragraph{Proof of Fact~\ref{fact:pq_interpolation}}
It remains to prove Fact~\ref{fact:pq_interpolation}, this is straightforward using Fact~\ref{fact:pvsq}.

\begin{proof}
        We start with the case $p \leq 2$.
        Let $v_2$ and $v_p$ be the maximizers of $\normtwoq{X}$ and $\normpq{X}$ respectively.
        Then it follows using the first case of Fact~\ref{fact:pvsq} twice that
        \begin{align*}
            \normtwoq{X} &= \frac{\normq{X v_2}}{\normt{v_2}} = \frac{\normq{X v_2}}{\norm{v_2}_p}\cdot \frac{\norm{v_2}_p}{\normt{v_2}} \leq d^{1/p - 1/2} \cdot \normpq{X}  \\
            &= d^{1/p - 1/2} \cdot \frac{\normq{X v_p}}{\norm{v_p}_p} \cdot \frac{\normt{v_p}}{\normt{v_p}} = d^{1/p - 1/2} \cdot \frac{\normq{X v_p}}{\norm{v_p}_2} \cdot \frac{\normt{v_p}}{\norm{v_p}_p} \\
            &\leq d^{1/p - 1/2} \cdot \normtwoq{X} \,.
        \end{align*}

        The $p \geq 2$ case is completely analogous using the second case of Fact~\ref{fact:pvsq}.
        \begin{align*}
            \normtwoq{X} &= \frac{\normq{X v_2}}{\normt{v_2}} = \frac{\normq{X v_2}}{\norm{v_2}_p}\cdot \frac{\norm{v_2}_p}{\normt{v_2}} \leq \normpq{X}  \\
            &= \frac{\normq{X v_p}}{\norm{v_p}_p} \cdot \frac{\normt{v_p}}{\normt{v_p}} = \frac{\normq{X v_p}}{\norm{v_p}_2} \cdot \frac{\normt{v_p}}{\norm{v_p}_p} \\
            &\leq d^{1/2 - 1/p} \cdot \normtwoq{X} \,.
        \end{align*}
    \end{proof}

\phantomsection
\addcontentsline{toc}{section}{References}
\bibliographystyle{amsalpha}
\bibliography{bib/custom,bib/dblp,bib/mathreview,bib/scholar}

\appendix
\clearpage

\section{Further Analysis of the Certificate of Guth, Maldague, and Urschel}
\label{sec:guth_limitation}
In this section we show that when $n \gtrsim d^3$ using the rows of the matrix and the top right singular vector of the covariance matrix does not certify a value better than $\Omega(d^{1/4})$.
This shows that the sparsification step is necessary to obtain a $d^{1/6}$-approximation via this list.
We remark that \cite{guth2025estimating} already contains an example showing that taking only this set of proxy directions cannot certify a better value than $n^{1/12}$, however in this example $n \approx d$ and thus this does not a priori rule out a scaling of $d^{1/12}$ for \emph{all} $n$.
We show that this is not the case.
Specifically, we show the following.
\begin{theorem}
    \label{thm:guth_limitation}
    Let $n \geq C d^3$ and $d \geq C$ for a large enough absolute constant $C$.
    There exists a matrix $X \in \R^{n \times d}$ with rows $x_1, \ldots, x_n$ and top right-singular vector $u$ such that
    \begin{enumerate}
    \item $\E_i \iprod{x_i,u}^4 \leq O(1)$.
    \item $\E_i \iprod{x_i, \bar{x}_j}^4 \leq O(1)$ for all $j = 1, \ldots, n$.
    \item There exists a unit vector $v$ such that $\E_i \iprod{x_i,v}^4 \geq d$.
\end{enumerate}
\end{theorem}

Our construction is probabilistic and we show that it succeeds with non-zero, in fact high, probability.
We emphasize that doing the error analysis carefully, instead of say just giving a proof sketch using `in-expectation' arguments, is important since we know the statement must be false when $n$ is much less than $d^3$.

Intuitively, a small fraction of the rows will be chosen to be very aligned with a specific direction $v$, say $e_1$, but since this fraction is small enough, we will only see this in the fourth moments instead of second moments.
We will also plant a (much smaller) spike in an orthogonal direction which will make sure that the top right singular vector of $X$ is orthogonal to $e_1$.
Indeed we can balance terms such that $\E_i x_i x_i^\top$ roughly looks like $I_d + e_2e_2^\top$. 
Further, by carefully choosing the norm of these rows, we can prevent the normalized row proxy directions to have large fourth moments. 

\begin{proof}

More formally, we do the following.
For simplicity assume that $r = n/d$ is an integer, it will be crucially important that $r \gtrsim d^2$.

\paragraph{The construction}
Let $Z_1, \ldots, Z_n \in \R^{d-1}$ be \iid samples from $N(0,I_{d-1} + e_1 e_1^\top)$.
Let $s_1, \ldots, s_r$ be independent (from the $Z_i$ and each other) and uniform over $\Set{-1,+1}$.
We set
\[
    x_i = \begin{cases}
        \sqrt{d}\cdot(s_i, d^{-1/4} Z_i) = (s_i \sqrt{d}, d^{1/4} Z_i)&\quad, i = 1, \ldots, r \,,\\
        (0, Z_i) &\quad, i = r+1, \ldots, n \,.
    \end{cases}
\]

It is easy to verify that $v = e_1$ induces a large four norm:
\begin{align*}
    \E_i \iprod{x_i,e_1}^4 \geq \tfrac 1 d \cdot \E_i \Brac{\iprod{x_i,e_1}^4 \mid i \leq r} = \tfrac 1 d \cdot d^2 = d \,.
\end{align*}

\paragraph{Basic observations}

We continue with some basic observations.
With overwhelming probability it will hold that for all $i$, $\normt{Z_i} = \Theta(\sqrt{d})$ and thus
\[
    \normt{x_i} = \begin{cases}
         \sqrt{d} \sqrt{1 + \Theta(d^{1/4})} = \Theta(d^{3/4})&\quad, i = 1, \ldots, r \,,\\
        \Theta(\sqrt{d}) &\quad, i = r+1, \ldots, n \,.
    \end{cases}
\]
In particular, it follows that
\[
    \E_i \normt{x_i}^4 \lesssim \tfrac 1 d \cdot d^{3} + \Paren{1 - \tfrac 1 d} \cdot d^2 \lesssim d^2 \,.
\]

The crucial observation, and where we use that $n > d^3$, is that the fourth moments of the distribution of $Z_i$ concentrate in every direction, even \emph{after} conditioning on $i \leq r$ or $i > r$.
Specifically, since $n \geq r \gtrsim d^2$ for any unit vector $v$, it holds that \cite{adamczak2010quantitative,vershynin2011approximating}, where $Z \sim N(0,I_{d-1} + e_2 e_2^\top)$
\[
    \E_i \Brac{\iprod{Z_i,v}^4 \mid i \leq r} \lesssim \E \iprod{Z,v}^4 \lesssim 1\,,
\]
and the same when we condition on $i > r$.
Note that the first concentration would not hold when $n$ is much smaller than $d^3$ (and hence $r \ll d^2$).

\paragraph{The top singular vector is close to $e_2$}

Let $u$ be the top singular vector of $X$.
We claim that with high probability, $u$ is very close to $e_2$.
Let $\hat{\Sigma} = \E_i x_i x_i^\top$.
Note that the top right singular vector is the same as the top eigenvector of $\hat{\Sigma}$.
We will show that with high probability there is a matrix $\Sigma$ such that
\begin{enumerate}
    \item $\normop{\hat{\Sigma}- \Sigma} \lesssim \sqrt{d/r}$,
    \item the top eigenvector of $\Sigma$ is equal to $e_2$,
    \item $\Sigma$ has a spectral gap of at least $0.99$.
\end{enumerate}
Standard perturbation arguments (see, e.g., \cite{yu2015useful}) then show that $\normt{u - e_2} \lesssim \sqrt{d/r}$.

Note that $x_i$ for $i > r$ are \iid distributed as $(0,Z)$ for $Z \sim N(0,I_{d-1})$.
Note that this distribution is $O(1)$-sub-Gaussian and has covariance matrix $\Sigma_2 = I_d - e_1 e_1 + e_2 e_2^\top$.
Let $\hat{\Sigma}_2 = \E_i \Brac{x_i x_i^\top \vert i > r}$.
Since $n-r \geq d$ it follows that (see, e.g., \cite{vershynin2018high})
\[
    \normop{\hat{\Sigma}_2 -\Sigma_2} \lesssim \sqrt{\frac{d}{n-r}} \lesssim \sqrt{\frac{d}{n}} \,.
\]

Similarly, $d^{-1/2} \cdot x_i$ for $i \leq r$ are \iid distributed as $(s, d^{-1/4} Z)$ where $Z$ is as before and $s$ is uniform over $\Set{-1,+1}$ and independent of $Z$.
This distribution is also $O(1)$-sub-Gaussian and has covariance matrix $\Sigma_1 = e_1 e_1^\top + \tfrac{1}{\sqrt{d}} \Sigma_2$.
Letting $\hat{\Sigma}_1 = \tfrac 1 d \cdot \E_i \Brac{x_i x_i^\top \vert i \leq r}$
Again, we obtain that $\normop{\hat{\Sigma}_1 - \Sigma_1} \lesssim \sqrt{d/r}$.

Finally, let $\Sigma = \Sigma_1 + (1-\tfrac 1 d) \Sigma_2$, then it follows by triangle inequality, that
\[
    \Normop{\hat{\Sigma} -  \Sigma} = \Normop{\hat{\Sigma}_1 + \Paren{1- \tfrac 1 d} \hat{\Sigma}_2 - \Sigma} \lesssim \sqrt{\frac d r} \,.
\]
In addition, since $\normop{\Sigma_2} \lesssim 1$ and by choosing $d$ large enough,
\begin{align*}
    \Normop{\Sigma - \Paren{I_d + e_2 e_2^\top}} &= \Normop{e_1 e_1^\top + \tfrac 1 {\sqrt{d}} \Sigma_2 + I_d - e_1 e_1^\top + e_2 e_2^\top - \tfrac 1 d \Sigma_2 - \Paren{I_d + e_2 e_2^\top}} \\
    &\leq \Normop{\tfrac 1 {\sqrt{d}} \Sigma_2 - \tfrac 1 d \Sigma_2}  \leq 0.01\,,
\end{align*}
which implies that the top eigenvector is $e_2$ and the spectral gap is at least $1/2$.

\paragraph{The top singular vector induces a small four norm}
We next use this to show that $\E_i \iprod{x_i, u}^4 \lesssim 1$.
Indeed,
\[
    \E_i \iprod{x_i, u}^4 \lesssim \E_i \iprod{x_i, e_2}^4 + \E_i \iprod{x_i, u-e_2}^4 \leq \E_i \iprod{x_i, e_2}^4  + \normt{u - e_2}^4 \cdot \E_i \normt{x_i}^4 \leq \E_i \iprod{x_i, e_2}^4 + \frac{d^4}{r^2} \,.
\]
The second term is $O(1)$ because $n \gtrsim d^3$ and hence $r^2 = (n/d)^2 \gtrsim d^4$.
For the first term, we bound
\begin{align*}
    &\tfrac{1}{d} \cdot \E_i \Brac{\iprod{x_i, e_2}^4 \mid i \leq r} + \Paren{1-\tfrac{1}{d}} \cdot \E_i \Brac{\iprod{x_i, e_2}^4 \mid i > r}  \\
    &= \tfrac 1 d \cdot d \cdot \E_i \Brac{Z_i(1)^4 \vert i \leq r} + \Paren{1-\tfrac{1}{d}} \cdot\E_i \Brac{Z_i(1)^4 \vert i > r}  \,,
\end{align*}
which is at most $O(1)$ by concentration of the empirical distributions of $Z_i$ conditioned on both $i \leq r$ and $i > r$.

\paragraph{Normalized rows induce small four norm}

We start with $j > r$.
It follows using concentration of the fourth moments of $Z_i$ conditioned on $i \leq r$ and $i > r$ that
\begin{align*}
    \E_i \iprod{x_i, \bar{x}_j}^4  &= \tfrac 1 d \cdot \E_i \Brac{\iprod{(s_i \sqrt{d}, d^{1/4} Z_i), (0,\bar{Z}_j)}^4 \mid i \leq r} + \Paren{1 - \tfrac{1}{d} } \cdot \E_i \Brac{\iprod{(0, Z_i), (0,\bar{Z}_j)}^4 \mid i > r} \\
    &\leq \E_i \Brac{\iprod{Z_i,\bar{Z}_j}^4 \mid i \leq r} + \E_i \Brac{\iprod{Z_i,\bar{Z}_j}^4 \mid i > r} \lesssim 1 + \frac{\norm{Z_j}^4}{n-r} \lesssim 1 + \frac{d^2}{n} \lesssim 1 \,.
\end{align*}

Next we consider $j \leq r$.
Similarly to above we know that $\E_i \brac{ \iprod{x_i, \bar{x}_j}^4 \mid i > r} \lesssim 1$.
Thus, it is enough to show the same for $\tfrac 1 d \cdot \E_i \brac{ \iprod{x_i, \bar{x}_j}^4 \mid i \leq r}$.
We can bound using $\norm{x_j} = \Theta(d^{3/4})$
\begin{align*}
    \tfrac 1 d \cdot \E_i \brac{ \iprod{x_i, \bar{x}_j}^4 \mid i \leq r} &= \tfrac 1 {d \normt{x_j}^4} \cdot \E_i \Brac{ \Paren{s_i s_j d + \sqrt{d} \iprod{Z_i,Z_j}}^4 \mid i \leq r} \\
    &\lesssim \tfrac 1 {d \normt{x_j}^4} \cdot \Paren{d^4 + d^2 \cdot \E_i \brac{  \iprod{Z_i,Z_j}^4 \mid i \leq r}} \\
    &\lesssim 1 + \tfrac 1 {d^2} \cdot \E_i \brac{  \iprod{Z_i,Z_j}^4 \mid i \leq r} \,.
\end{align*}
The second part we can bound using $r \geq d^2$as
\[
    \frac{1}{d^2} \cdot \Paren{\frac{\normt{Z_j}^8}{r} + \normt{Z_j}^4} \lesssim 1 \,.
\]

\end{proof}

\section{Various Missing Pieces}
\label{sec:explanations}
\subsection{Further Analysis of the Baseline Certificate}
\label{sec:baseline}

We will show the following theorem.
The upper bound follows using the same approach discussed in \cref{sec:intro}.
\begin{theorem}
    \label{thm:baseline}
    Let $q \in \N$ be even and $\alpha > 0$ and let $X \in \R^{n \times d}$ satisfy $\normtwoqexp{X} \leq \alpha$. Then $\normop{\E_i x_i^{\otimes q/2} [x_i^{\otimes q/2}]^\top} \leq \alpha^q \cdot d^{q/4}$.
    Further, when $q$ is divisible by 4\footnote{This assumption is mostly for convenience and we believe that the same lower bound holds for all even $q$.} and $n \gtrsim d^{q/2}$ there exist matrices $X$ such that $\normtwoqexp{X} \leq O_q(1)$ but $\normop{\E_i x_i^{\otimes q/2} [x_i^{\otimes q/2}]^\top} \geq \Omega_q(d^{q/4})$.
\end{theorem}

\begin{proof}
    To show the first part, we upper bound the spectral norm by the Frobenius norm.
    Recall that in the proof of \cref{thm:two_to_q_full} we showed that $\normtwoqexp{X} \leq \alpha$ implies that $\E_i \normt{x_i}^q \leq \alpha^q d^{q/2}$.
    Using this, it follows that
    \begin{align*}
        \Normop{\E_i x_i^{\otimes q/2} [x_i^{\otimes q/2}]^\top} &\leq \Normf{\E_i x_i^{\otimes q/2} [x_i^{\otimes q/2}]^\top} = \sqrt{\E_{i,j} \iprod{x_i,x_j}^q} = \sqrt{\E_i \normt{x_i}^q \cdot \E_j \iprod{\bar{x}_i, x_j}^q} \\
        &\leq \sqrt{\alpha^q \cdot \E_i \normt{x_i}^q} \leq \alpha^q \cdot d^{q/4} \,.
    \end{align*}

    To show the second part, we take the rows of $X$ to be \iid standard Gaussian.
    Then by standard concentration bounds it follows that with high probability \cite{adamczak2010quantitative,vershynin2011approximating}
    \[
        \sup_{\normt{v} = 1} \E_i \iprod{x_i,v}^q \lesssim \E_{Z \sim N(0,I_d)} \iprod{Z,v}^q \leq q^{q/2} = O_q(1) \,.
    \]
    Further, since $q$ is divisible by $4$ we can consider the test vector $\mathrm{Vec}(I_d^{\otimes q/4})$ which has Frobenius norm $d^{q/8}$.
    It follows that
    \begin{align*}
        \Normop{\E_i x_i^{\otimes q/2} [x_i^{\otimes q/2}]^\top} \geq d^{-q/4} \cdot \E_i \iprod{x_i^{\otimes q/2}, \mathrm{Vec}(I_d^{\otimes q/4})}^2 = d^{-q/4} \cdot \E_i \normt{x_i}^q \,.
    \end{align*}
    The last expectation is $\Omega_q(d^{q/2})$ with high probability, which proves the lower bound.
\end{proof}

\subsection{Models in Algorithmic Statistics}
\label{sec:alg_stats_models}

For a distribution with covariance $\Sigma$, the Mahalanobis norm is a vector-norm and refers to $\normt{\Sigma^{-1/2} (\cdot)}$.
The relative spectral norm is a matrix-norm and refers to $\normop{\Sigma^{-1/2}(\cdot )\Sigma^{-1/2}}$.

`Robust estimation' refers to the fact that an adversary can arbitrarily change an $\e$-fraction of the input data set, as described for robust mean estimation.
In the list-decodable model only a subset $S$ of size $\alpha n$ of the input satisfies the distributional assumption, the rest can be arbitrary.
This makes estimating the mean of $S$ impossible.
Instead, we aim to output a short list of vectors which contains one vector that is close to the mean of $S$.
In a Mixture Model, the samples can be grouped into $k$ separated components; for simplicity we assume that all groups have the same size.
In this problem, the relevant parameter is how large this separation needs to be so that we can faithfully cluster the data.
In all other problems, the success metric is estimating the relevant parameter up to small error.

\subsection{Tensoring Sum-of-Squares Lower Bounds}
\label{sec:tensoring_sos_lower_bound}

In this section, we give a short argument proving \cref{obs:SoS_lower_bound}.
It is based on the sum-of-squares lower bounds for constant factor approximations in \cite{harrow2019limitations} and a tensoring/parallel repetition argument in \cite{bhattiprolu2023inapproximability}.

\restateobs{obs:SoS_lower_bound}

\begin{proof}
    \cite[Corollary 4.9]{harrow2019limitations} shows that for $k$ and $N$ large enough, there is a matrix $X \in \R^{N \times k}$ and constants $s < c$ such that $\normtwofourexp{X} \leq s$, but the value of the degree $\log k / \poly (\log \log k)$ sum-of-squares relaxation is at least $c$.
    That is, there is a pseudo-expectation over unit vectors such that $\pE_v \E_i \iprod{x_i,v}^4 \geq c^4$.
    For simplicity, assume that $\ell = \tfrac{\log_{c/s} d}{p}$ is an integer.
    Set $n = N^\ell, d = k^\ell$ and consider the $n \times d$ matrix $\tilde{X}$ that has all rows of the form $x_{i_1} \otimes \ldots x_{i_\ell}$ for $i_j \in [N]$.
    By \cite[Theorem 4.11]{bhattiprolu2023inapproximability} it follows that $\normtwofourexp{\tilde{X}} \leq \normtwofourexp{X}^\ell \leq s^\ell$.

    Next, consider the pseudo-distribution $\pE'$ obtained by taking the $\ell$-fold product of $\pE$, i.e., with variables $v' \coloneqq v^{(1)} \otimes \ldots \otimes v^{(\ell)}$, where $v^{(j)}$ are independent copies of the variables of $\pE$.
    Note that the degree of $\pE'$ is still
    \[
        \frac {\log k} {\poly(\log \log k )} = \frac{\log (d) / \ell}{\poly (\log [\log (d) / \ell])} = \frac{p}{\poly(\log p)}\,.
    \]
    This is also a pseudo-distribution over unit vectors since $\normt{v'}^2 = \normt{v^{(1)}}^2 \cdot \ldots \cdot \normt{v^{(\ell)}}^2$.
    Lastly, it holds that
    \begin{align*}
        \pE'_{v'} \E_r \iprod{\tilde{X}_r, v'}^4 &= \pE_{v^{(1)}, \ldots, v^{(\ell)}} \E_{i_1, \ldots, i_\ell} \iprod{x_{i_1} \otimes \ldots \otimes x_{i_\ell}, v^{(1)} \otimes \ldots \otimes v^{(\ell)}}^4 = \pE_{v^{(1)}, \ldots, v^{(\ell)}} \E_{i_1, \ldots, i_\ell} \prod_{j=1}^\ell \iprod{x_{i_j}, v^{(j)}}^4 \\
        &= \prod_{j=1}^\ell \pE_{v^{(j)}} \E_{i_j} \iprod{x_{i_j}, v^{(j)}}^4 \geq c^{4\ell} \,.
    \end{align*}
    Thus, the gap of the tensored instance is at least $\Paren{\tfrac c s}^\ell = d^{1/p}$.

\end{proof}

\end{document}